\documentclass{tlp}
\usepackage{aopmath,bbm, amsfonts}
\usepackage{color}

\newtheorem{definition}{Definition} 
\newtheorem{example}{Example} 

\def \R{\mbox{$\mathbb R$}}

\usepackage{latexsym}
\usepackage{psfrag,graphicx,paralist}

\newcommand{\diag}[2]{d_{{#1}{#2}}}

\newcommand{\rarrow}{\rightarrow}

\newcommand{\la}{\langle}
\newcommand{\ra}{\rangle}
\newcommand{\success}{{\bf success}}

\newcommand{\os}{[\![}
\newcommand{\cs}{]\!]}
\newcommand{\rrarrow}{\longrightarrow}
\newcommand{\Rrarrow}{\longmapsto}

\newcommand{\doo}{{\bf do}}
\newcommand{\watching}{{\bf watching}}
\newcommand{\elsee}{{\bf else}}
\newcommand{\noow}{{\bf now}}
\newcommand{\thene}{{\bf then}}
\newcommand{\tell}{{\bf tell}}

\newcommand{\ask}{{\bf ask}}

\newcommand{\timeout}{{\bf timeout}}

\def\ent{\vdash}
\def\0{{\mathbf 0}}
\def\1{{\mathbf 1}}
\def\C{{\mathcal C}}
\long\def\comment#1{}

\begin{document}
\bibliographystyle{acmtrans}

\long\def\comment#1{}

\title[TSCCP: An Interleaved and a Parallel Approach]{Timed Soft Concurrent Constraint Programs: An Interleaved and a Parallel Approach}

\author[Bistarelli et al.]
{Stefano Bistarelli \\
Dipartimento di Matematica e Informatica, Universit\`a di Perugia\\
Via Vanvitelli 1, 06123 Perugia Italy\\
E-mail: vista@dmi.unipg.it
\and
Maurizio Gabbrielli\\
Dipartimento di Scienze dell'Informazione, Universit\`a di Bologna\\
Via Zamboni 33, 40126 Bologna, Italy\\
E-mail: gabbri@cs.unibo.it
\and
Maria Chiara Meo\\
Dipartimento di Economia, Universit\`a ``G. D'Annunzio''\\
Viale Pindaro 42, 65127 Pescara, Italy\\
E-mail: cmeo@unich.it
\and
Francesco Santini\\
Centrum Wiskunde \& Informatica (CWI)\\
Science Park 123, 1098XG Amsterdam, The Netherlands\\
E-mail: F.Santini@cwi.nl
}

\submitted{20 May 2012}
\revised{15 February 2014}
\accepted{18 February 2014}


\maketitle

\label{firstpage}

\begin{abstract}
We propose a timed and soft extension of Concurrent Constraint
Programming. 
The time extension is based on the hypothesis of {\itshape bounded
asynchrony}: the computation takes a bounded period of time and is
measured by a discrete global clock. Action prefixing is then
considered as the syntactic marker which distinguishes a time
instant from the next one. Supported by soft constraints instead
of crisp ones, {\itshape tell} and {\itshape ask} agents are now
equipped with a preference (or consistency) threshold which is
used to determine their success or suspension. In the paper we
provide a language to describe the agents behavior, together with
its operational and denotational semantics, for which we also
prove the compositionality and correctness properties. After presenting a semantics using maximal parallelism of actions, we also describe a version for their interleaving on a single processor (with maximal parallelism for time elapsing). Coordinating agents
that need to take decisions both on preference values and time events may benefit from this language. To appear in Theory and Practice of Logic
Programming (TPLP).

\end{abstract}


\section{Introduction}\label{sec:intro}
Time is a particularly important aspect of cooperative
environments. In many ``real-life'' computer applications, the
activities have a temporal duration (that can be even interrupted)
and the coordination of such activities has to take into
consideration this timeliness property. The interacting actors are
mutually influenced by their actions, meaning that $A$ reacts
accordingly to the timing and quantitative aspects related to $B$'s behavior,
and vice versa. In fact, these interactions can be often related
to quantities to be measured or minimized/maximized, in order to
take actions depending from these scores: consider, for example,
some generic communicating agents that need to take decisions on a (monetary) cost or a (fuzzy) preference for a shared resource. They
both need to coordinate through time-dependent  and preference-based decisions.

A practical example of such agents corresponds, for example, to software agents that need to negotiate some service-level agreement on a resource, or a service, with  time-related side-conditions. For instance, a fitting example is given by auction schemes, where the seller/bidder agents need to agree on a preference for a given prize (e.g., a monetary cost). At the same time, the agents have to respect some timeout and alarm events, respectively representing the absence and the presence of bids for the prize (for instance).  The language we present in this paper is well suited for this kind of interactions, as 
Section~\ref{sec:mpexample} shows with examples.

The {\itshape Timed Concurrent Constraint Programming}
({\itshape tccp}), a timed extension of the pure formalism of
{\itshape Concurrent Constraint Programming} ({\itshape
ccp})~\cite{Sa89a}, has been introduced in~\cite{BGM00}. The language is based on the
hypothesis of {\em bounded asynchrony}~\cite{SJG96}: computation takes a bounded period of time
rather than being instantaneous as in the concurrent synchronous
languages {\tt ESTEREL}~\cite{BG92}, {\tt LUSTRE}~\cite{HCP91},
{\tt SIGNAL}~\cite{GBGM91} and Statecharts~\cite{Ha87}. Time
itself is measured by a discrete global clock, i.e., the internal
clock of the {\itshape tccp} process. In~\cite{BGM00} the authors also
introduced {\em timed reactive sequences}, which describe the reaction of a {\itshape tccp} process to the
input of the external environment, at each
moment in time. Formally, such a reaction is a
pair of constraints $\langle c,d\rangle$, where $c$ is the input
 and $d$ is the constraint produced by the
process in response to  $c$. \comment{(due to the monotonicity of {\itshape
ccp} computations, $c$ includes always the input).}

Soft constraints~\cite{bistarellibook,jacm} extend classical
constraints to represent multiple consistency levels, and thus
provide a way to express preferences, fuzziness, and uncertainty.
The {\itshape ccp} framework has been extended to work with soft
constraints~\cite{scc}, and the resulting framework is named
{\itshape Soft Concurrent Constraint Programming} ({\itshape
sccp}). With respect to {\itshape ccp}, in {\itshape sccp} the
{\em tell} and {\em ask} agents are equipped with a preference (or
consistency) threshold, which is used to determine their success,
failure, or suspension, as well as to prune the search; these
preferences should preferably be satisfied but not necessarily
(i.e. over-constrained problems). We adopt soft constraints
instead of crisp ones, since classic constraints
 show evident limitations  when trying
to represent real-life scenarios, where the knowledge is not
completely available nor crisp.

In this paper, we introduce a timed and soft extension of {\itshape
ccp} that we call {\itshape Timed Soft Concurrent Constraint
Programming} ({\itshape tsccp}), inheriting from both {\itshape
tccp} and {\itshape sccp} at the same time.
In {\itshape tsccp}, we directly introduce  a timed interpretation of the usual
programming constructs of \emph{sccp}, by identifying a time-unit with
the time needed for the execution of a basic \emph{sccp} action (ask and tell),
and  by interpreting action prefixing as the
next-time  operator.
An explicit timing primitive is also
introduced in order to allow for the specification of timeouts. In the first place,
the parallel operator of {\itshape tsccp} is first interpreted  in terms of maximal parallelism, as in \cite{BGM00}.
Secondly, we also consider a different paradigm, where the parallel operator is interpreted
in terms of interleaving, however assuming maximal parallelism
for actions depending on time. In other words, time passes for
all the parallel processes involved in a computation. This
approach, analogous to that one adopted in \cite{BoGaMe04}, is different from that one of \cite{BGM00,coord} (where
maximal parallelism was assumed for any kind of action), and it
is also different from the one considered in \cite{BGZ00}, where time does not elapse for timeout constructs.
This can be accomplished by allowing all the time-only dependent actions ($\tau$-transitions)
to concurrently run with at most one action manipulating the store (a $\omega$-transition).

The paper extends the results in  \cite{coord} by providing  new semantics  that allows
maximal parallelism for time elapsing  and an
interleaving model for basic computation steps (see Section~\ref{sec:interleaving}). This new language is called \emph{{\itshape tsccp} with interleaving}, i.e., {\itshape tsccp-i}, to distinguish it from the version allowing maximal parallelism of all actions. According to the maximal parallelism policy (applied, for example, in the original works as \cite{Sa89a} and \cite{saraswat2}), at each moment every enabled agent of the system is activated, while in the interleaving paradigm only one of the enabled agents is executed instead. This second paradigm is more realistic if we consider limited resources, since it does not imply the existence of an unbounded number of processors. However, in \cite{BGM00} it is shown that the notion of maximal parallelism of {\itshape tsccp} is more expressive than the notion of interleaving parallelism of other concurrent constraint languages. The presence of maximal parallelism can force the computation to discard some (non-enabled) branches which could became enabled later on (because of the information produced by parallel agents), while this is not possible when considering an interleaving model. Therefore, {\itshape tsccp} is sensitive to delays in adding constraints to the store, whereas this is not
the case for {\itshape ccp} and {\itshape tsccp-i}.

The rest of the paper is organized as follows: in Section~\ref{sec:background} we summarize the most important
background notions and frameworks from which {\itshape tsccp}
derives, i.e. {\itshape tccp} and {\itshape sccp}. In
Section~\ref{sec:tsccp} we present the {\itshape tsccp} language, and in Section~\ref{sec:mpopsem}  describes the operational semantics
of  {\itshape tscc} agents. Section~\ref{sec:mpexample} better explains the programming idioms as \emph{timeout} and \emph{interrupt},
exemplifies  the use
of timed paradigms in  the {\itshape tscc} language and  shows an
application example on modeling an auction interaction among several bidders and a single auctioneer.  Section~\ref{sec:mpdensem}  describes the  denotational semantics  for {\itshape tsccp}, and  proves the
denotational model correctness with the aid of {\itshape connected
reactive sequences}.  Section~\ref{sec:interleaving} explains the semantics for interleaving with  maximal parallelism of time-elapsing actions (i.e. the {\itshape tsccp-i} language), while Section~\ref{sec:timeline} describes a timeline for the execution of three parallel agents in {\itshape tsccp-i}.
Section~\ref{sec:compdensemtsccpi}  describes the  denotational semantics  of {\itshape tsccp-i} and proves the correctness of the
denotational model. Section~\ref{sec:related} reports the related work and, at last,  Section~\ref{sec:conclusions}
concludes by also indicating future
research.

\section{Background}\label{sec:background}

\subsection{Soft Constraints}
\label{sec:scsp}

A {\em soft constraint}~\cite{jacm,bistarellibook} may be seen as
a constraint where each instantiation of its variables has an
associated value from a partially ordered set which can be
interpreted as a set of preference values. Combining constraints
will then have to take into account such additional values, and
thus the formalism has also to provide suitable operations for
combination ($\times$) and comparison ($+$) of tuples of values
and constraints. This is why this formalization is based on the
concept of \emph{c-semiring}~\cite{jacm,bistarellibook}, called just semiring in the rest of the paper.

\paragraph{Semirings.}
A semiring is a tuple $\langle A,+,\times,\0,\1 \rangle$ such
that: {\em i)} $A$ is a set and $\0, \1 \in A$; {\em ii)} $+$ is
commutative, associative and $\0$ is its unit element; {\em iii)}
$\times$ is associative, distributes over $+$, $\1$  is its unit
element and $\0$ is its absorbing element. A c-semiring is a
semiring $\langle A,+,\times,\0,\1 \rangle$ such that: $+$ is
idempotent, $\1$ is its absorbing element and $\times$ is
commutative.
Let us consider the relation $\leq_S$ over $A$ such that $a \leq_S
b$ iff $a+b = b$. Then, it is possible to prove that
(see~\cite{jacm}): {\em i)} $\leq_S$ is a partial order; {\em ii)}
$+$ and $\times$ are monotone on $\leq_S$; {\em iii)} $\0$ is its
minimum and $\1$ its maximum; {\em iv)} $\langle A,\leq_S \rangle$
is a complete lattice (a complete lattice is a partially ordered set in which all subsets have both a supremum and an infimum) and, for all $a, b \in A$, $a+b = \mathit{lub}(a,b)$
(where $\mathit{lub}$ is the {\em least upper bound}).

Moreover, if $\times$ is idempotent, then: $+$ distributes over
$\times$; $\langle A,\leq_S \rangle$ is a complete distributive
lattice and $\times$ its $glb$ ({\em greatest lower bound}).
Informally, the relation $\leq_S$ gives us a way to compare
semiring values and constraints. In fact, when we have $a \leq_S
b$, we will say that {\em $b$ is better than $a$}. In the following,
when the semiring will be clear from the context, $a \leq_S b$
will be often indicated by $a \leq b$.
\paragraph{Constraint System.}
Given a semiring $S = \langle A,+,\times,\0,\1 \rangle$ and an
ordered set of variables $V$ over a finite domain $D$, a {\em soft
constraint} is a function which, given an assignment $\eta :
V\rightarrow D$ of the variables, returns a value of the semiring.
Using this notation  
$\C = \eta \rightarrow A$ is the set of all possible constraints
that can be built starting from $S$, $D$ and $V$.

Any function in $\C$ involves all the variables in $V$, but we
impose that it depends on the assignment of only a finite subset
of them. So, for instance, a binary constraint $c_{x,y}$ over
variables $x$ and $y$, is a function $c_{x,y}: (V\rightarrow
D)\rightarrow A$, but it depends only on the assignment of
variables $\{x,y\}\subseteq V$ (the {\em support} of the
constraint, or {\em scope}).
Note that $c\eta[v:=d_1]$ means $c\eta'$ where $\eta'$ is $\eta$
modified with the assignment $v:=d_1$ (that is the operator $[\ ]$
has precedence over application). Note also that $c\eta$ is the
application of a constraint function $c:(V \rightarrow D)
\rightarrow A$ to a function $\eta:V\rightarrow D$; what we
obtain, is a semiring value $c\eta$.

The partial order $\leq_S$ over $\C$ can be easily extended among
constraints by defining $c_1 \sqsubseteq c_2 \iff c_1 \eta \leq
c_2 \eta$, for each possible $\eta$.

\paragraph{Combining and projecting soft constraints.}
Given the set $\C$, the combination function $\otimes: \C\times\C
\rightarrow \C$ is defined as $(c_1\otimes c_2)\eta =
c_1\eta\times c_2\eta$ (see also~\cite{jacm,bistarellibook,scc}).
Informally, performing the $\otimes$ between two constraints means
building a new constraint whose support involves all the variables
of the original ones, and which associates with each tuple of
domain values for such variables a semiring element which is
obtained by multiplying the elements associated by the original
constraints to the appropriate sub-tuples. 

Given a constraint $c \in \C$ and a variable $v \in V$, the {\em
projection}~\cite{jacm,bistarellibook,scc} of $c$ over $V-\{v\}$,
written $c\Downarrow_{(V-\{v\})}$ is the constraint $c'$ s.t.
$c'\eta = \sum_{d \in D} c \eta [v:=d]$.
Informally, projecting means eliminating some variables from the
support. This is done by associating with each tuple over the
remaining variables a semiring element which is the sum of the
elements associated by the original constraint to all the
extensions of this tuple over the eliminated variables. 

We define also a function $\bar{a}$~\cite{bistarellibook,scc} as
the function that returns the semiring value $a$ for all
assignments $\eta$, that is, $\bar{a}\eta =a$. We will usually
write $\bar{a}$ simply as $a$.
An example of constants that will be useful later are $\bar{\0}$
and $\bar{\1}$ that represent respectively the constraints
associating $\0$ and $\1$ to all the assignment of domain values.

\paragraph{Solutions.} A SCSP~\cite{bistarellibook} is defined as $P= \langle V, D, C, S \rangle$,
where $C$ is the set of constraints defined over variables in $V$ (each with domain $D$), and whose preference is determined by semiring $S$. The {\em best level of consistency}
notion is defined as $\mathit{blevel}(P) = \mathit{Sol}(P) \Downarrow_{\emptyset}$,
where $Sol(P)= \bigotimes C$~\cite{bistarellibook}. A problem $P$ is $\alpha$-consistent if
$\mathit{blevel}(P) = \alpha$~\cite{bistarellibook}. $P$ is instead simply
``consistent'' iff there exists $\alpha
>_S \0$ such that $P$ is $\alpha$-consistent. $P$ is inconsistent
if it is not consistent.

\begin{figure}[h]
\centering
\includegraphics[scale=0.6]{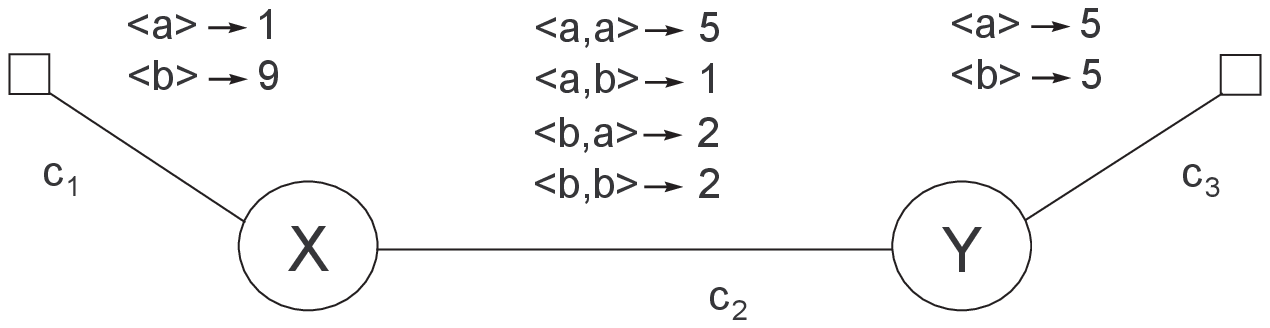} 
\caption{A SCSP based on a weighted semiring.}
\label{figure:wexample}
\end{figure}

\begin{example}Figure~\ref{figure:wexample} shows a
weighted SCSP as a graph: the weighted semiring is used, i.e.
$\langle \R^{+} \cup \{\infty\}, \min, \hat{+},$ $\infty,
0\rangle$ ($\hat{+}$ is the arithmetic plus operation). Variables
and constraints are represented respectively by nodes and arcs
(unary for $c_1$-$c_3$, and binary for $c_2$); $D= \{a, b\}$. The
solution of the CSP in Figure~\ref{figure:wexample} associates a
semiring element to every domain value of variables $X$ and $Y$ by
combining all the constraints together, i.e. $Sol(P)= \bigotimes C$.
For instance, for the tuple $\langle a, a\rangle$ (that is, $X = Y =
a$), we have to compute the sum of $1$ (which is the value assigned
to $X = a$ in constraint $c_1$), $5$ ($\langle X = a, Y = a \rangle$ in $c_2$) and $5$ ($Y = a$ in $c_3$): the value for this tuple is
$11$.  The solution $X = a, Y = b$ is a $7$-consistent solution, where $7$ corresponds to the  \emph{blevel} of $P$, i.e., $\mathit{Sol}(P)\Downarrow_\emptyset = 7$.
\end{example}

\subsection{Concurrent Constraint Programming over Soft Constraints}\label{sec:sccp}

\label{sec:semf}

The basic idea underlying {\itshape ccp}~\cite{Sa89a} is that
computation progresses via monotonic accumulation of information
in a constraint global store. Information is produced by the
concurrent and asynchronous activity of several agents which can
add ({\em tell}) a constraint to the store. Dually, agents can
also check ({\em ask}) whether a constraint is entailed by the
store, thus allowing synchronization among different agents. The
{\itshape ccp} languages are defined parametrically w.r.t. a given
{\em constraint system}. The notion of constraint system has been
formalized in~\cite{SR90} following Scott's treatment of
information systems. Soft constraints over a semiring $S = \langle
\mathcal{A},+,\times,\0,\1 \rangle$ and an ordered set of
variables $V$ (over a domain $D$) have been showed to form a constraint system {\em ``\`{a} la
Saraswat''}, thus leading to the definition of \emph{Soft Concurrent Constraint Programming}g (\emph{sccp})~\cite{jacm,bistarellibook,scc}.

Consider the  set $\C$ and the partial order $\sqsubseteq$. Then
an entailment relation $\ent \subseteq \wp(\C) \times \C$ is
defined s.t. for each $C \in \wp(\C)$ and $c \in \C$, we have $C
\ent c \iff \bigotimes C \sqsubseteq c$ (see
also~\cite{bistarellibook,scc}). Note that in this setting the
notion of token (constraint) and of set of tokens (set of
constraints) closed under entailment is used indifferently. In
fact, given a set of constraint functions $C_1$, its closure
w.r.t. entailment is a set $\bar{C_1}$ that contains all the
constraints greater than $\bigotimes C_1$. This set is univocally
representable by the constraint function $\bigotimes C_1$. The
definition of the entailment operator $\ent$ on top of $\C$, and of
the $\sqsubseteq$ relation, lead to the notion of {\em soft
constraint system}. It is also important to notice that
in~\cite{Sa89a} it is claimed that a constraint system is a {\em
complete} {\em algebraic} lattice. In the \emph{sccp} framework,
algebraicity is not required~\cite{scc} instead, since the algebraic nature
of the structure $\C$ strictly depends on the properties of the
semiring\footnote{Notice that we do not aim at computing the
closure of the entailment relation, but only to use the entailment
relation to establish if a constraint is entailed by the current
store, and this can be established even if the lattice is not
algebraic (that is even if the times operator is not
idempotent).}.

To treat the hiding operator of the language, a general notion of
existential quantifier is introduced by using notions similar to
those used in cylindric algebras. Consider a set of variables $V$
with domain $D$ and the corresponding soft constraint system $\C$.
For each $x \in V$, the hiding function~\cite{bistarellibook,scc}
is the function $(\exists_x c)\eta =\sum_{d_i\in D} c\eta[x :=
d_i]$. To make the hiding operator computationally tractable, it
is required that the number of domain elements in $D$, having
semiring values different from $\0$, is finite~\cite{scc}. In this
way, to compute the sum needed for $(\exists_x c)\eta$, we can consider just a finite number of elements
(those different from $\0$), since $\0$ is the unit element of the
sum.  Note that
by using the hiding function we can represent the $\Downarrow$
operator defined in Section~\ref{sec:scsp}. In fact, for any
constraint $c$ and any variable $x \subseteq V$,
$c\Downarrow_{V-x} = \exists_x c$~\cite{scc}.

To model parameter passing also diagonal elements have to be
defined. Consider a set of variables $V$ and the corresponding
soft constraint system. Then, for each $x,y \in V$, a diagonal
constraint is defined as $d_{xy} \in \C$
s.t., $d_{xy}\eta[x:= a, y := b]= \1$ if $a=b$, and $d_{xy}\eta[x:=
a, y := b] = \0$ if $a\neq b$~\cite{bistarellibook,scc}.

\begin{theorem}[cylindric constraint system~\cite{scc}]
Consider a semiring $S = \langle \mathcal{A},+,\times,\0,\1
\rangle$, a domain of the variables $D$, an ordered set of
variables $V$, and the corresponding structure $\C$. Then,
$S_C=\langle \C, \otimes,\0,\1 , \exists_x, d_{xy}\rangle$, is a
cylindric constraint system.
\end{theorem}

\subsection{Timed Concurrent Constraint Programming}
\label{sec:tccp}
A  timed extension of {\itshape ccp}, called {\itshape tccp} has
been introduced in~\cite{BGM00}. Similarly to other existing timed
extensions of {\itshape ccp} defined in~\cite{SJG96}, {\itshape
tccp} is a language for reactive programming designed around the
hypothesis of {\em bounded asynchrony} (as introduced
in~\cite{SJG96}: computation takes a bounded period of time rather
than being instantaneous).

When querying the store for some information that is not present
(yet), a {\itshape ccp}  agent will simply suspend until the
required information has arrived. In timed applications however
often one cannot  wait indefinitely for an event. Consider for
example the case of a connection to a web service
providing some on-line banking facility. In case the connection cannot
be established, after a reasonable amount of time an appropriate time-out message
has to be communicated to the user.  A timed
language should then allow us to specify that, in case a given
time bound is exceeded (i.e. a {\em time-out} occurs), the wait is
interrupted and an alternative action is taken.
Moreover, in some cases it is also necessary to have a preemption mechanism
which allows one to abort an active process $A$ and to
start a process $B$ when a specific (abnormal) event occurs.

In order to be able to specify these timing constraints
 {\itshape tccp} introduces a discrete global clock and assumes
that {\em ask} and {\em tell} actions take one time-unit.
Computation evolves in steps of one time-unit, so called
clock-cycles. Action prefixing is the syntactic marker which
distinguishes a time instant from the next one and it is assumed
that parallel processes are executed on different processors,
which implies that, at each moment, every enabled agent of the
system is activated. This assumption gives rise to what is called
{\em maximal parallelism}. The time in between two successive
moments of the global clock intuitively corresponds to the
response time of the underlying constraint system. Thus all
parallel agents are synchronized by the response time of the
underlying constraint system. Since the store is monotonically increasing and one can have dynamic process creation, clearly the previous assumptions imply that the constraint solver takes a constant time (no matter how big the store is), and that there is an unbounded number of processors.
However, one can impose suitable restriction on programs, thus ensuring that the (significant part of the) store and the number of processes do not exceed a fixed bound; these restrictions would still allow significant forms of recursion with parameters.

Furthermore,
a timing construct of the form $\noow  \ c \ \thene\
A \ \elsee \ B $ is introduced in  {\itshape tccp}, whose semantics is the following: if
the constraint $c$ is entailed by the store at the current time
$t$, then the above agent behaves as $A$ at time $t$, otherwise it
behaves as $B$ at time $t$. This
basic construct allows to derive such timing mechanisms as
time-out and preemption~\cite{BGM00,SJG96}. The instantaneous reaction can be obtained by evaluating \noow $c$ in parallel with  $A$ and $B$, within the same time-unit. At the end of this time-unit, the store will be updated by using either the constraint produced by $A$, or that one produced by  $B$, depending on the result of the evaluation of \noow $c$. Clearly, since $A$ and $B$ could contain nested  $\noow \ \thene \ \elsee$  agents, a limit for the number of these nested agents should be fixed. Note that, for recursive programs, such a limit is ensured by the presence of the procedure-call, since we assume that the evaluation of such calls takes one time-unit.

\section{Timed Soft Concurrent Constraint Programming}
\label{sec:tsccp}

In this section we present the {\itshape tsccp} language, which
originates from both  {\itshape tccp} and  {\itshape sccp}. To
obtain this aim, we extend the syntax of the {\itshape cc}
language with the timing construct \noow $c$ \thene $A$
\elsee $B$ (inherited from {\itshape tccp}), and also in order
to directly handle the cut level as in {\itshape sccp}. This means
that the syntax and semantics of the {\bf tell}, {\bf ask} and
\noow agents have to be enriched with a threshold that is used to check
when the agents may succeed, or suspend.

\begin{definition}[{\itshape tsccp} Language]\label{def:tscclanguage}
Given a soft constraint system $\langle S,D,V\rangle$, the
corresponding structure $\C$, any semiring value $a$, {\em soft constraints}
$\phi , \ c \in {\cal C}$ and any tuple of variables $x$,  the syntax of the {\itshape tsccp}
language is given by the following grammar:
\[
\begin{array}{ll}
P ::= & F \text{.} A\\
F ::= & p(x):: A \;|\; F.F\\
A ::= & {\bf success} \;|\; \hbox{\tell}(c)\rarrow_\phi A
\;|\;\hbox{\tell}(c)\rarrow ^a A \;|\; E \;|\; A\parallel A
\;|\;
 \exists x A \;|\; p(x) \;|\;  \\
 & \Sigma_{i=1}^{n} E_i\;|\; \noow_\phi \ c \ \thene\
A \ \elsee \ A \;|\;  \noow^a \ c \ \thene\
A \ \elsee \ A\\
E ::= &\hbox{\bf ask}(c) \rarrow_\phi A  \;|\; \hbox{\bf ask}(c)
\rarrow^a A \\
\end{array}
\]
where, as usual, $P$ is the class of processes, $F$ is the class
of sequences of procedure declarations (or clauses), $A$ is the
class of agents. In a {\itshape tsccp} {\em process} $P=F\text{.}A$,
$A$ is the initial agent, to be executed in the context of the set of declarations $F$. The agent {\bf success} represents a successful termination, so it may not make any  further transition.

\end{definition}

In the following, given an agent $A$,  we denote by $Fv(A)$ the
set of the free variables of $A$ (namely, the variables which do
not appear in the scope of the $\exists$ quantifier).
Besides the use of soft constraints (see Section~\ref{sec:semf}) instead of crisp ones, there are two fundamental differences between {\itshape tsccp} and {\itshape ccp}.
The first main difference w.r.t. the original {\em cc} syntax is
the presence of a semiring element $a$ and of a constraint $\phi$
to be checked whenever an {\em ask} or {\em tell} operation is
performed. More precisely, the level $a$ (respectively, $\phi$) will be
used as a cut level to prune computations that are not good
enough. The second main difference with respect to {\itshape ccp }  (but, this time,
also with respect to {\itshape sccp}) is instead the presence of the \noow $c$ \thene\ $A$ \elsee\ $B$ construct introduced in
Section~\ref{sec:tccp}. Even for this construct, the level $a$
(or $\phi$) is used as a cut level to prune computations.

Action prefixing is denoted by $\rightarrow$, non-determinism is
introduced via the guarded choice construct $\Sigma_{i=1}^{n} E_i$, parallel
composition is denoted by $\parallel$, and a notion  of locality
is introduced by the agent $\exists x A$, which behaves like $A$
with $x$ considered local to $A$, thus hiding the information on
$x$ provided by the external environment.

In the next subsection
we formally describe the operational semantics of {\itshape
tsccp}. In order to simplify the notation, in the following
we will usually write a {\itshape tsccp} {\em process} $P=F\text{.}A$
 simply as the corresponding agent $A$.

\section{An Operational Semantics for {\itshape tsccp} Agents}\label{sec:mpopsem}

The operational model of {\itshape tscc} agents can be formally
described by a transition system $T= ({\it Conf}, \rrarrow )$
where we assume that each transition step takes exactly one
time-unit. Configurations in {\itshape Conf} are pairs
consisting of a process and of  a constraint in ${\cal C}$,
representing the common {\em store} shared by all the agents. The transition relation
$\rrarrow \subseteq  {\it Conf} \times {\it Conf}$ is the least
relation satisfying the rules {\bf R1-R17} in Figure~\ref{mpt1}, and it
characterizes the (temporal) evolution of the system. So, $\langle
A,\sigma\rangle \rrarrow \langle B,\delta\rangle $ means that, if
at time $t$ we have the process $A$ and the store $\sigma$, then at
time $t+1$ we have the process $B$ and the store $\delta$.

\begin{figure}

    \begin{center}

\begin{tabular}{llll}
&\mbox{   }&\mbox{   } &\mbox{   }
\\
\mbox{\bf R1}& $\frac {\displaystyle (\sigma \otimes
c)\Downarrow_{\emptyset} \not< a}{\displaystyle
\begin{array}{l}
\la \hbox{\tell}(c) \rightarrow^{a} A, \sigma \ra \rrarrow \la A,
\sigma \otimes c\ra
\end{array}}$\ \ \ & \bf{V-tell}&
\\
&\mbox{   }&\mbox{   } &\mbox{   }
\\
\mbox{\bf R2}& $\frac {\displaystyle \sigma \otimes c
\not\sqsubset \phi}{\displaystyle
\begin{array}{l}
\la \hbox{\tell}(c)\rightarrow_{\phi} A, \sigma \ra \rrarrow \la
A, \sigma \otimes c\ra
\end{array}}$ & \bf{Tell} &
\\
&\mbox{   }&\mbox{   } &\mbox{   }
\\
\mbox{\bf R3}& $\frac {\displaystyle \sigma \ent c \ \ \ \ \
\sigma\Downarrow_{\emptyset} \not< a}{\displaystyle
\begin{array}{l}
\la \hbox{\ask}(c) \rightarrow^{a} A, \sigma \ra \rrarrow \la A,
\sigma \ra
\end{array}}$\ \ \ & \bf{V-ask}&
\\
&\mbox{   }&\mbox{   } &\mbox{   }
\\
\mbox{\bf R4}& $\frac {\displaystyle \sigma \ent c \ \ \ \  \sigma
\not\sqsubset \phi}{\displaystyle
\begin{array}{l}
\la \hbox{\ask}(c) \rightarrow_{\phi} A, \sigma  \ra \rrarrow \la
A, \sigma \ra
\end{array}}$ & \bf{Ask}&
\\
&\mbox{   }&\mbox{   } &\mbox{   }
\\

\mbox{\bf R5 }& $\frac {\displaystyle \la A,\sigma \ra \rrarrow
\la A', \sigma \otimes \delta  \ra\ \ \ \ \la B,\sigma\ra \rrarrow
\la B', \sigma \otimes \delta' \ra } {\displaystyle
\begin{array}{l}
\la A\parallel B,\sigma\ra\rrarrow \la A'\parallel B',\sigma
\otimes \delta\otimes \delta'\ra
\end{array}}$ & \bf{Parall1}&
\\
&\mbox{   }&\mbox{   }&\mbox{   }
\\
\mbox{\bf R6 }& $\frac {\displaystyle \la A,\sigma\ra \rrarrow \la
A', \sigma'\ra\ \ \ \ \la B,\sigma \ra \not\rrarrow}
{\displaystyle
\begin{array}{l}
\la A\parallel B, \sigma \ra\rrarrow \la A'\parallel B, \sigma'
\ra
\\
\la B\parallel A, \sigma \ra\rrarrow \la B\parallel A', \sigma'
\ra
\end{array}}$&
 \bf{Parall2}&
\\
&\mbox{   }&\mbox{   }&
\\

\mbox{\bf R7}& $\frac {\displaystyle \la  E_j,\sigma\ra \rrarrow
\la A_j,\sigma' \ra\ \ \ \ \ \ j\in [1,n]} {\displaystyle
\begin{array}{l}
\la \Sigma_{i=1}^{n}E_i , \sigma \ra\rrarrow \la A_j,\sigma'\ra
\end{array}}$ & \bf{Nondet}&
\\
&\mbox{   }&\mbox{   }&
\\
\mbox{\bf R8 }& $\frac {\displaystyle \la A,\sigma \ra \rrarrow
\la A', \sigma' \ra \ \ \ \ \ \sigma \ent c \ \ \ \ \
\sigma\Downarrow_{\emptyset} \not< a}
 {\displaystyle
\begin{array}{l}
\la \noow^{a} \ c \ \thene \ A\ \elsee \
B,\sigma\ra\rrarrow \la A',  \sigma' \ra
\end{array}}$  & \bf{V-now1}&
\\
&\mbox{   }&\mbox{   }&
\\
\mbox{\bf R9 }& $\frac {\displaystyle \la A,\sigma \ra \not
\rrarrow \ \ \ \ \ \sigma \ent c \ \ \ \ \
\sigma\Downarrow_{\emptyset} \not< a} {\displaystyle
\begin{array}{l}
\la \noow^a \ c \ \thene \ A\ \elsee \
B,\sigma\ra\rrarrow \la A,  \sigma \ra
\end{array}}$  & \bf{V-now2}&
\\
&\mbox{   }&\mbox{   }&
\\
\mbox{\bf R10 }& $\frac {\displaystyle \la B,\sigma \ra \rrarrow
\la B', \sigma'  \ra \ \ \ \ \ \sigma \not \ent c \ \ \ \
\sigma\Downarrow_{\emptyset} \not< a} {\displaystyle
\begin{array}{l}
\la \noow^a \ c \ \thene \ A\ \elsee \
B,\sigma\ra\rrarrow \la B',  \sigma '\ra
\end{array}}$& \bf{V-now3}&
\\
&\mbox{   }&\mbox{   }&
\\
\mbox{\bf R11 }& $\frac {\displaystyle \la B,\sigma\ra
\not\rrarrow\ \ \ \ \ \sigma \not \ent c \ \ \ \ \
\sigma\Downarrow_{\emptyset} \not< a
 }
{\displaystyle
\begin{array}{l}
\la \noow^a \ c \ \thene \ A\ \elsee \
B,\sigma\ra\rrarrow \la B,  \sigma\ra
\end{array}}$& \bf{V-now4}&
\\
&\mbox{   }&\mbox{   }&
\\
\mbox{\bf R12 }& $\frac {\displaystyle \la A,\sigma \ra \rrarrow
\la A', \sigma' \ra  \ \ \ \ \sigma \ent c \ \ \ \ \sigma
\not\sqsubset \phi} {\displaystyle
\begin{array}{l}
\la \noow_\phi \ c \ \thene \ A\ \elsee \
B,\sigma\ra\rrarrow \la A',  \sigma' \ra
\end{array}}$  & \bf{Now1}&
\\
&\mbox{   }&\mbox{   }&
\\
\mbox{\bf R13 }& $\frac {\displaystyle \la A,\sigma \ra \not
\rrarrow \ \ \ \ \sigma \ent c \ \ \ \ \sigma \not\sqsubset \phi}
{\displaystyle
\begin{array}{l}
\la \noow_\phi \ c \ \thene \ A\ \elsee \
B,\sigma\ra\rrarrow \la A,  \sigma \ra
\end{array}}$  & \bf{Now2}&
\\
&\mbox{   }&\mbox{   }&
\\
\mbox{\bf R14 }& $\frac {\displaystyle \la B,\sigma\ra \rrarrow
\la B', \sigma'  \ra \ \ \ \ \sigma \not \ent c \ \ \ \ \sigma
\not\sqsubset \phi} {\displaystyle
\begin{array}{l}
\la \noow_\phi \ c \ \thene \ A\ \elsee \
B,\sigma\ra\rrarrow \la B', \sigma'\ra
\end{array}}$& \bf{Now3}&
\\
&\mbox{   }&\mbox{   }&
\\
\mbox{\bf R15 }& $\frac {\displaystyle \la B,\sigma\ra
\not\rrarrow \ \ \ \ \sigma \not \ent c \ \ \ \  \sigma
\not\sqsubset \phi} {\displaystyle
\begin{array}{l}
\la \noow_\phi \ c \ \thene \ A\ \elsee \
B,\sigma\ra\rrarrow \la B,  \sigma\ra
\end{array}}$& \bf{Now4}&
\\
&\mbox{   }&\mbox{   }&
\\
\mbox{\bf R16}& $\frac {\displaystyle \la A[x/y],  \sigma
\ra\rrarrow\la B, \sigma' \ra} {\displaystyle\la \exists x A,
\sigma \ra\rrarrow\la B, \sigma' \ra}$
&\bf{Hide} &\\
&\mbox{   }&\mbox{   }&
\\
\mbox{\bf R17}& $\la p(x),\sigma\ra\rrarrow\la A,
\sigma\ra \ \ \ \ {\it p(x) :: A \in F}$
&\bf{P-call}&\\
&\mbox{   }&\mbox{   }&
\\
\end{tabular}
    \end{center}
  \caption{The transition system for {\itshape tsccp}.}\label{mpt1}

\end{figure}

Let us now briefly discuss the rules in Figure~\ref{mpt1}. Here is a
brief description of the transition rules:

\begin{description}
  \item[Valued-tell.]
The valued-tell rule checks for the $a$-consistency
  of the {\itshape Soft Constraint Satisfaction
  Problem}~\cite{bistarellibook} (SCSP)
  defined by the store $\sigma \otimes c$. A SCSP $P$ is
$a$-consistent if $blevel(P) = a$, where $blevel(P) =
Sol(P) \Downarrow_{\emptyset}$, i.e., the {\em best level of
consistency} of the problem $P$ is a semiring value representing
the least upper bound among the values yielded by the solutions.
Rule ${\bf R1}$ can
  be applied only if the store $\sigma \otimes c$ is $b$-consistent
  with $b \not< a$\footnote{Notice that we use $b \not< a$ instead of $b \geq a$ because
  we can possibly deal with partial orders. The same holds also for $\not \sqsubset$ instead of $\sqsupseteq$.}. In this case the agent evolves to the new
  agent $A$ over the store $\sigma \otimes c$.
  Note that different
  choices of the {\em cut level} $a$ could possibly lead to different
  computations.
Finally, note that the updated store $\sigma \otimes c$ will be
visible only starting from the next time instant, since each
transition step involves exactly one time-unit.

  \item[Tell.] The tell action is a finer check of the
  store. In this case (see rule {\bf R2}), a pointwise comparison between the store
  $\sigma \otimes c$ and the constraint $\phi$
  is performed. The idea is to perform an overall
  check of the store, and to continue the computation only if there is
  the possibility to compute a solution not worse
  than $\phi$.
  Note that this notion of tell could be also applied to the classical {\itshape cc} framework:
 the tell operation would succeed when the set of tuples satisfying constraint $\phi$
 is not a superset of the set of tuples allowed by $\sigma \cap c$.\footnote{Notice that the $\otimes$ operator
  in the crisp case reduces to set intersection.}
  As for the valued tell, the updated store $\sigma \otimes c$ will be
visible only since the next time instant.
In the following, let us use $\hbox{\tell}(c)\rightarrow A $  and $\hbox{\tell}(c)$ as a shorthand for $\hbox{\tell}(c)\rightarrow_{\bar{\0}}A$
and $\hbox{\tell}(c)\rightarrow_{\bar{\0}}{\bf success}$, respectively.

  \item[Valued-ask.] The semantics of the valued-ask is extended in
  a way similar to what we have done for the valued-tell action. This
  means that, to apply the rule {\bf R3}, we need to check if the store $\sigma$
  entails the constraint $c$, and also if $\sigma$ is ``consistent
  enough'' w.r.t. the threshold $a$ set by the programmer.

  \item[Ask.] In rule {\bf R4}, we check if the store $\sigma$
  entails the constraint $c$, but, similarly to  rule {\bf R2}, we also compare a finer (pointwise)
  threshold $\phi$ to the store $\sigma$.
  As for the tell action, let us use $\hbox{\ask}(c) \rightarrow A$ as a shorthand for $\hbox{\ask}(c) \rightarrow_{\bar{\0}}A$.

\item[Parallelism.] Rules {\bf R5} and {\bf R6} model the parallel
composition operator in terms of {\em maximal parallelism}: the
agent $A\parallel B$ executes in one time-unit all the initial
enabled actions of $A$ and $B$. Considering rule {\bf R5} (where {\em maximal parallelism} is accomplished in practice), notice
that the ordering of the operands in $\sigma \otimes \delta\otimes
\delta'$ is not relevant, since $\otimes$ is commutative and
associative. Moreover, for the same two properties, if
$\sigma \otimes \delta= \sigma \otimes \gamma$ and $\sigma \otimes
\delta'= \sigma \otimes \gamma'$, we have that $\sigma \otimes
\delta \otimes \delta'= \sigma \otimes \gamma \otimes \gamma'$.
Therefore the resulting store $\sigma \otimes \delta \otimes
\delta'$ is independent from the choice of the constraint $\delta$
such that $\la A,\sigma \ra \rrarrow \la A', \sigma'  \ra$ and
$\sigma '=\sigma \otimes \delta$ (analogously for $\delta'$).

  \item[Nondeterminism.] According to rule ${\bf R7}$, the
guarded choice operator gives rise to global non-determinism: the
external environment can affect the choice, since $\ask(c_j)$ is
enabled at time $t$ (and $A_j$ is started at time $t+1$) if and
only if the store $\sigma$ entails $c_j$ (and if it is compatible with the
threshold too), and $\sigma$ can be modified by other agents.

\item[Valued-now and Now.] Rules ${\bf R8}$-${\bf R11}$ show that the agent
$\noow^a \ c \ \thene \ A\ \elsee\ $ $B$
behaves as $A$ or $B$ depending on the fact
that $c$ is or is not entailed by the store, provided that the current store $\sigma$ is compatible with the
threshold. Differently from the
case of the ask, here the evaluation of the guard is
instantaneous: if current store $\sigma$ is compatible with the
threshold $a$, $\langle A, \sigma\rangle$ ($\langle B,
\sigma\rangle$) can make a transition at time $t$ and $c$ is (is
not) entailed by the store $\sigma$, then the agent $\noow^a \ c \ \thene \ A\ \elsee \ B$ can make the same transition at time $t$.
Moreover, observe that in any case the control is passed either to
$A$ (if $c$ is entailed by the current store $\sigma$ and $\sigma$ is compatible with the
threshold) or to $B$
(in case $\sigma$ does not entail $c$ and $\sigma$ is compatible with the
threshold). Analogously for the not-valued version, i.e.,  $\noow_\phi \ c \ \thene \ A\ \elsee \ B$ (see rules ${\bf R12}$-${\bf R15}$).
Finally, we use $\noow \ c \ \thene \ A\ \elsee \ B$ as a shorthand for the agent $\noow _{\bar{\0}}\ c \ \thene \ A\ \elsee \ B$

  \item[Hiding variables.] The agent $\exists x A$ behaves like $A$,
  with $x$ considered {\itshape local} to $A$, as show by rule {\bf R16}. This is obtained by
  substituting the variable $x$ for a variable $y$, which we assume
  to be new and not used by any other process. Standard renaming techniques
  can be used to ensure this; in rule {\bf R16}, $A[x/y]$ denotes the process obtained
  from $A$ by replacing the variable $x$ for the variable $y$.

\item[Procedure-calls.] Rule ${\bf R17}$ treats the case of a
procedure-call when the actual parameter equals the formal
parameter. We do not need more rules since, for the sake of
simplicity, here and in the following we assume that the set
{\itshape F} of procedure declarations is closed w.r.t. parameter
names: that is, for every procedure-call $p(y)$ appearing in a
process {\itshape F\text{.}A}, we assume that, if the original declaration
for {\itshape p} in {\itshape F} is $p(x) :: A$, then {\itshape F}
contains also the declaration $p(y) :: \exists x (\hbox{\tell}(\diag{x}{y})
\parallel A)$.\footnote{Here the (original) formal parameter is
identified as a local alias of the actual parameter.
Alternatively, we could have introduced a new rule treating
explicitly this case, as it was in the original {\itshape ccp}
papers.} Moreover, we assume that if $p(x) :: A \in F$, then
$Fv(A) \subseteq x$.

\end{description}

Using the transition system described by (the rules in)
Figure~\ref{mpt1}, we can now define our notion of observables, which
considers  the results
of successful terminating computations that the agent $A$ can
perform for each {\itshape tsccp} process $P= F\text{.}A$. \\
Here and in the following,  given a transition relation $\rrarrow$, we denote by $\rrarrow ^*$ its reflexive and transitive closure.

\begin{definition}[Observables]
Let $P= F\text{.}A$ be a {\itshape tsccp} process. We define
$${\cal O}^{mp}_{io}(P) = \{
\gamma \Downarrow_{Fv(A)} \mid \la A, \1 \ra\rrarrow ^*\la {\bf
Success}, \gamma\ra\}$$

\end{definition}
where ${\bf Success}$ is any agent which contains only occurrences
of the agent ${\bf success}$ and of the operator $\parallel$.

\section{Programming Idioms and Examples}\label{sec:mpexample}
We can consider the primitives in Definition~\ref{def:tscclanguage} to derive the soft
version of the programming idioms in~\cite{BGM00}, which are
typical of reactive programming.

\begin{itemize}
    \item [{\it Delay}.] The delay constructs $\tell(c)
    \stackrel{t}{\longrightarrow}_\phi A$ or
    $\ask(c)\stackrel{t}{\longrightarrow}_\phi A$ are used to delay the
    execution of agent $A$ after the execution of $\tell(c)$ or $\ask(c)$; $t$ is the number of
    the time-units of delay. Therefore, in addiction to a
    constraint $\phi$, in {\itshape tsccp} the
    transition arrow can have also a number of delay
    slots. This idiom can be defined by induction: the base case
    is $\stackrel{0}{\longrightarrow}_\phi A \equiv \rightarrow_\phi A$, and the inductive step is
$\stackrel{n+1}{\longrightarrow}_\phi A\equiv
\rightarrow_\phi \tell({\bar{\1}})\stackrel{n}{\longrightarrow}_{\bar{\0}} A$. The valued version can be defined in an analogous way.
    \item [{\it Timeout}.] The timed guarded choice agent $\Sigma_{i=1}^{n} \ask(c_i) \rightarrow_i A_i \, \timeout(m) \, B$
waits at most $m$ time-units ($m \geq 0$) for the satisfaction of one of the guards;
notice that all the ask actions have a soft transition arrow, i.e.  $\rightarrow_i$ is either of the form  $\rightarrow_{\phi_i}$ or  $\rightarrow^{a_i}$, as in Figure~\ref{mpt1}. Before this time-out, the process behaves just like the guarded choice: as soon as there exist enabled guards, one of them (and the corresponding branch) is nondeterministically selected. After waiting for $m$ time-units, if no guard is enabled, the timed choice agent behaves as $B$.
Timeout constructs can be assembled through the composition of several $\noow_{\phi} \ c  \ \thene \ A \  \elsee \ B$ primitives (or their  valued version),  as explained in \cite{BGM00} for the (crisp) \emph{tccp} language.

The timeout can be defined inductively as follows:  let us denote by $A$ the agent $\Sigma_{i=1}^{n} \ask(c_i) \rightarrow_i A_i$. In the base case, that is $m=0$, we define $\Sigma_{i=1}^{n} \ask(c_i) \rightarrow_i A_i \, \timeout(0) \, B$ as the agent:

$$
\begin{array}{llll}
  \noow_1 \ c_1 & \hspace*{-0.2cm}\thene \ A &  &  \\
   & \hspace*{-0.75cm}\elsee \ (\ \noow_2 \ c_2 &\hspace*{-0.2cm}\thene \ A &   \\
   &  & \hspace*{-0.75cm}\elsee \ (\dots (\ \noow_n \ c_n \ \thene \ A\ \elsee \ \ask(\bar{\1}) \rightarrow \ B) \dots))
\end{array}
$$

where for $i=1, \ldots,n$, either $\noow_i = \noow_{\phi_i}$ if  $\rightarrow_i$ is of the form  $\rightarrow_{\phi_i}$
or $\noow_i = \noow^{a_i}$ if  $\rightarrow_i$ is of the form $\rightarrow^{a_i}$. Because of the operational semantics explained in rules {\bf R8-R11} (see Figure~\ref{mpt1}), if a guard $c_{i}$ is  true, then the agent $\Sigma_{i=1}^{n} \ask(c_i) \rightarrow_i A_i$  is evaluated in the same time slot. Otherwise, if no guard $c_i$ is true, the agent $B$ is evaluated in the next time slot.
Then, by inductively reasoning on the number of time-units $m$, we can define $\Sigma_{i=1}^{n} \ask(c_i) \rightarrow_i A_i \, \timeout(m) \, B$ as
$$\Sigma_{i=1}^{n} \ask(c_i) \rightarrow_i A_i \, \timeout(0) \, (\Sigma_{i=1}^{n} \ask(c_i) \rightarrow_i A_i \, \timeout(m-1) \, B)\text{.}$$

\item [{\it Watchdog}.] Watchdogs are used to interrupt the activity of a process on a signal
from a specific event. The idiom $\doo \;  A \; \watching_{\phi} \;c$ behaves as $A$, as long as $c$ is not
entailed by the store and the current store is compatible with the
threshold; when $c$ is entailed and the current store is compatible with the
threshold, the process $A$ is
immediately aborted.

The reaction is instantaneous, in the sense
that $A$ is aborted at the same time instant of the detection of
the entailment of $c$. However, according to the computational model, if $c$ is detected
at time $t$, then $c$ has to be produced at time $t'$ with $t'<t$. Thus, we have a form of
weak preemption.

As well as timeouts, also watchdog agents can be defined in terms of the other basic constructs of the language (see Figure~\ref{fig:mpwatchex}).

In the following we assume that there exists an (injective) renaming function $\rho$
which, given a procedure name $p$, returns a new name $\rho( p)$ that is not used elsewhere in the program.
Moreover, let us use $\noow_{\phi} \, c \ \elsee \ B$ as a shorthand for $\noow_{\phi}\, c \ \thene \ {\bf success} $ \elsee$ \ B$, where we assume that, for any procedure $p$ declared as $p(x) ::A$, a declaration
$\rho(p)(x) :: \doo \; \rho(A) \; \watching_\phi \; c$ is added, where $\rho(A)$ denotes the agent obtained
from $A$ by replacing in it each occurrence of any procedure $q$ by $\rho(q)$. The
assumption in the case of the $\exists x A$ agent is needed for correctness. In practical cases,
it can be satisfied by suitably renaming the variables associated to signals.
In the following $\rightarrow'$ is either of the form  $\rightarrow_{\psi}$ or  $\rightarrow^{a}$.
Analogously for $\noow'$.

\begin{figure}[!h]
\scalebox{0.8}{
\hspace{-1.2cm}\begin{tabular}{lll}
\mbox{ }&\mbox{ }&\\
$\doo \; {\bf success} \; \watching_{\phi} \ c $ & $\Longrightarrow$ & ${\bf success} $\\
\mbox{ }&\mbox{ }&\\
$\doo \; \tell(d)\rightarrow' A \; \watching_{\phi} \ c $ & $\Longrightarrow$ & $\noow_{\phi} \ c \ \elsee \
 \tell(d)\rightarrow' \doo \;  A \; \watching_{\phi} \ c $\\
\mbox{ }&\mbox{ }&\\

$\doo \; \Sigma_{i=1}^{n} \ask(c_i) \rightarrow_i A_i \; \watching_{\phi} \;c\ $ & $\Longrightarrow$ & $\noow_{\phi} \ c \ \elsee \ \Sigma_{i=1}^{n} \ask(c_i)\rightarrow_i \doo \; A_i \; \watching_{\phi} \; c$\\

\mbox{ }&\mbox{ }&\\
$\doo \ (\noow' \ d \ \thene \ A \ \elsee \ B) \ \watching _{\phi}\;c$ & $\Longrightarrow$ & $
\hspace*{-0.2cm}\begin{array}[t]{lll}
  \noow' \ d & \thene \  \doo \; A \; \watching_{\phi}  \, c \\
   & \elsee \ \doo \; B \ \watching_{\phi}  \ c
\end{array}
$\\
\mbox{ }&\mbox{ }&\\

$\doo \; A \parallel B \; \watching_{\phi} \;c$ & $\Longrightarrow$ & $\doo \;  A \; \watching_{\phi} \;c \parallel \doo \;B \; \watching_{\phi}  \;c$\\
\mbox{ }&\mbox{ }&\\

$\doo \; \exists x A \; \watching_{\phi} \;c$  & $\Longrightarrow$ & $ \exists  x \;\doo  \; A \; \watching_{\phi} \;c$, assuming
$\exists  _xc=c$ \\
\mbox{ }&\mbox{ }&\\

$\doo  \; p(x)  \; \watching_{\phi} \;c$  & $\Longrightarrow$ & $\noow_{\phi} \ c \ \elsee \
\rho(p)(x) \; \watching_{\phi} \ c $\\
\mbox{ }&\mbox{ }&\\

\end{tabular} }
\caption{Examples of watchdog constructs.}\label{fig:mpwatchex}
\end{figure}

The translation in Figure~\ref{fig:mpwatchex} can be easily extended to the case of the agent
$\doo \;  A \; \watching_{\phi} \;c \; \elsee \; B$, which behaves as the previous watchdog and also activates the process $B$
when $A$ is aborted (i.e., when $c$ is entailed and the current state is compatible with the
threshold). In the following we will then use also
this form of watchdog.

The assumption on the instantaneous evaluation of $\noow_{\phi} \, c $ is essential in order
to obtain a preemption mechanism which can be expressed in terms of the $\noow_{\phi} \ \thene \ \elsee$  primitive. In fact, if the evaluation of $\noow_{\phi} \ c$  took one time-unit, then this
unit delay would change the compositional behavior of the agent controlled by the
watchdog. Consider, for example, the agent $A= \tell(a)\rightarrow \tell(b)$, which takes two
time-units to complete its computation. The agent $A^t= \noow \ c \ \elsee \
 \tell(a)\rightarrow  \noow \ c \ \elsee \
 \tell(b)$
 (resulting from the translation of $\doo \;  A \; \watching_{\bar{\0}} \;c$) compositionally behaves  as $A$, unless a $c$ signal is detected and the current state is compatible with the threshold, in which case the evaluation of $A$ is interrupted. On the other hand, if the evaluation of $\noow \ c$ took one time-unit, then $A^t$
would take four time-units and would not behave anymore as $A$ when $c$ is not present. In fact, in this case, the agent $A \parallel B$ would produce $d$ while $A^t \parallel B$ would not,
where $B$ is the agent $\ask(\bar{\1}) \rightarrow \noow \ a \ \thene \ \tell(d)\ \elsee \ {\bf success}$.

\noindent The valued version of watchdogs can be defined in an analogous way.
\begin{figure}
\begin{center}
$c_{1}:(\{x\} \rightarrow  \mathit{N}) \rightarrow \mathit{R}^{+}
 \; \; \text{ s.t. } c_{1}(x) = x + 3 \hspace{0.65cm} c_{2}:(\{x\} \rightarrow
\mathit{N}) \rightarrow \mathit{R}^{+} \; \; \text{ s.t. } c_{2}(x) =
x+5$
 \end{center}
\begin{center} $c_{3}:(\{x\} \rightarrow \mathit{N}) \rightarrow \mathit{R}^{+} \; \;
\text{ s.t. } c_{3}(x) = 2x + 8$\hspace{1cm} \end{center}\caption{Three (weighted) soft
constraints; $c_{3}= c_{1} \, \otimes \,
c_{2}$, $c_{2} \ent c_{1}$, $c_{3} \ent c_{1}$  and $c_{3} \ent c_{2}$. }\label{fig:mpconstexample}
\end{figure}

\end{itemize}

With this small set
of idioms, we have now enough expressiveness to describe complex
interactions. For the following examples on the new programming idioms, we consider the \emph{Weighted} semiring $\langle
\mathbbm{R}^{+} \cup\{+\infty\},min,+,+\infty,0 \rangle$~\cite{bistarellibook,jacm} and the (weighted) soft constraints in Figure~\ref{fig:mpconstexample}.
We first provide simple program examples in order to explain as more details as possible on how a computation of {\itshape tsccp} agents proceeds. In Section~\ref{sec:mpauctionexample} we show a more complex example describing the classical actions during a negotiation process; the aim of that example is instead to show the expressivity of the {\itshape tsccp} language, without analyzing its execution in detail.

\begin{example}[Delay]\label{exmp1} As a first very simple example, suppose to have two  agents $A_1, A_2$ of the form:
$A_1:: \tell(\bar{\1}) \stackrel{2}{\longrightarrow}\!^{+\infty} \,\tell(c_{2}) \rightarrow\!^{+\infty} \,\success$ and $A_2:: \tell(\bar{\1}) \stackrel{1}{\longrightarrow}\!^{+\infty} \,\ask(c_1)\rightarrow^{9} \, \success$; their concurrent evaluation in the $\bar{\1} \equiv \bar{0}$ empty store is:
$$ \langle (\tell(\bar{0}) \stackrel{2}{\longrightarrow}\!^{+\infty}\, \tell(c_{2}){\rightarrow}\!^{+\infty} \,{\success}) \parallel (\tell(\bar{0}) \stackrel{1}{\longrightarrow}\!^{+\infty} \, \ask(c_1)\rightarrow^9 \,{\success}) , \bar{0} \rangle\textit{.}$$

The timeline for this parallel execution is described in Figure~\ref{figure:mptimeline1}. For the evaluation of $\tell$ and $\ask$ we respectively consider the rules ${\bf R1}$ and ${\bf R3}$ in Figure~\ref{mpt1}, since both transitions are $a$-valued. However, both these two actions are delayed: three time-units for the $\tell(c_2)$ of $A_1$ (including the first $\tell(\bar{0})$), and two time-units for the $\ask(c_1)$ of $A_2$ (including the first $\tell(\bar{0})$). As explained before, this can be obtained by  adding $\bar{\1}$ to the store with a $\tell$ action respectively three, and two times. Therefore, the parallel agent $A_1 \parallel A_2$ corresponds to:
$$  (\tell(\bar{0} )\rightarrow\!^{+\infty} \,\tell(\bar{0} )\rightarrow\;^{+\infty}\,\tell(\bar{0} )\rightarrow\!^{+\infty} \,\tell(c_2)\rightarrow\!^{+\infty}
\,\success )\parallel$$
$$(\tell(\bar{0} )\rightarrow\!^{+\infty}\,\tell(\bar{0} )\rightarrow\!^{+\infty}\,\ask(c_1) \rightarrow ^9 \, \success)\textit{.}$$

This agent is interpreted by using ${\bf R5}$-${\bf R6}$ in Figure~\ref{mpt1} in terms of maximal parallelism, i.e., all the actions are executed in parallel. The first two $\tell$ of $A_1$ and $A_2$ can be simultaneously executed by using rule ${\bf R1}$: the precondition $(\bar{0} \otimes \bar{0}) \Downarrow_\emptyset= 0 \not< 9$ of the rule is then satisfied. The store does not change since $\bar{0} \otimes \bar{0} = \bar{0}$.
At this point, the $\ask$ action of $A_2$ is not enabled because $\bar{0} \not\ent c_1$, that is the precondition $\sigma \ent c_1$ of ${\bf R3}$ is not satisfied. Therefore, the processor can only be allocated to $A_1$ and, since $(\bar{0} \otimes \bar{0}) \Downarrow_\emptyset= 0 \not< +\infty$ is true (i.e. the precondition of ${\bf R1}$ is satisfied), at $t = 3$ the computation is in the state:
$$ \langle \tell(c_{2})\rightarrow\!^{+\infty} \, \success \parallel \ask(c_1) \rightarrow^{9} \, \success, \bar{0} \rangle\textit{.}$$

Now the $\tell$ can be executed because $(\bar{0} \otimes c_2) \Downarrow_\emptyset= 5 \not< +\infty$: therefore, the store becomes equal to $\bar{0} \otimes c_2 = c_2$:
$$ \langle  \success \parallel \ask(c_1) \rightarrow^{9}\, \success, c_2 \rangle\textit{.}$$

At $t = 5$ (see Figure~\ref{figure:mptimeline1}) we can successfully terminate the program: in the   store $\sigma = c_2$ the $\ask$ is finally enabled at $t=4$, according to the two preconditions of rule ${\bf R3}$, i.e., $c_2 \ent c_1$ and $c_2 \Downarrow_\emptyset= 5 \not< 9$: therefore we have $A_1 \parallel A_2:: \langle  \success \parallel  \success, c_2 \rangle\textit{.}$

\begin{figure}[h]
\centering
\includegraphics[scale=1.2]{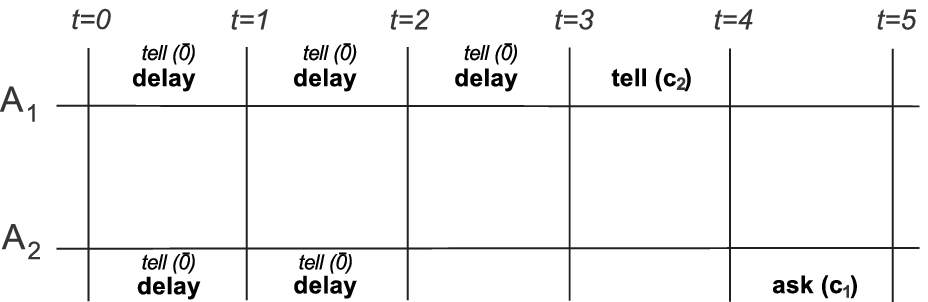} 
\caption{The timeline of the execution of the $A_1 \parallel A_2$ parallel agent in Example~\ref{exmp1}.}
\label{figure:mptimeline1}
\end{figure}
\end{example}

\begin{example}[Timeout]\label{exmp2}  In this second example we evaluate a timeout construct. Suppose we have two agents
 $A_1$ and $A_2$ of the form: $$A_1 :: ((\ask(c_1)\rightarrow\!^{+\infty} \,\success) \,+\, (\ask(c_2)\rightarrow\!^{+\infty}
\,\success)) \; \timeout(1)$$ $$\ask(c_1)\rightarrow\!^{+\infty} \,\success $$ and
$$A_2 :: \tell(\bar{0})\stackrel{2}{\longrightarrow}\!^{+\infty} \,\tell(c_3) \rightarrow\!^{+\infty} \,\success$$

The description of  agent $A_1$ is a shortcut for the following agent, as previously explained in the definition of the timeout:
$$\noow^{+\infty} \ c_1 \ \thene \ B \ \elsee \ (\noow^{+\infty} \ c_2 \ \thene \ B   \
   \elsee \ $$ $$ (\ask(\bar{1}) \rightarrow \
   \noow^{+\infty} \ c_1 \ \thene \ B \ \elsee \ $$
$$ (\noow^{+\infty} \ c_2 \ \thene \ B \ \elsee \
(\ask(\bar{1}) \rightarrow \ask(c_1)\rightarrow\!^{+\infty} \,\success))))\textit{.}$$
where $B::(\ask(c_1)\rightarrow\!^{+\infty}\,\success + \ask(c_2)\rightarrow\!^{+\infty} \,\success)$. Their concurrent evaluation in the $\bar{\1} \equiv \bar{0}$ empty store is:
 $$\begin{array}{ll}
     \langle (B \; \timeout(1) \;
     \ask(c_1)\rightarrow\!^{+\infty} \,\success \parallel \\
     \tell(\bar{0})\stackrel{2}{\longrightarrow}\!^{+\infty}\,
   \tell(c_3)\rightarrow\!^{+\infty} \,\success), \bar{0}\rangle\textit{.}
 \end{array}
$$

The timeline for this parallel execution is given in Figure~\ref{figure:mptimeline2}. At $t=0$ the store is empty (i.e., $\sigma= \bar{0}$), thus both  constraints $c_1$ and $c_2$ asked by the nondeterministic choice agent $A_1$ are not entailed. In $A_2$, the $\tell$ of $c_3$, which would entail both $c_1$ and $c_2$, is delayed by three time-units: in the first  three time-units, $\tell(\bar{0} )\rightarrow^{+\infty}$ is executed  according to the delay construct, as shown in Example~\ref{exmp1}. At $t=2$ the timeout is triggered in $A_1$, since, according to ${\bf R1}$, ${\bf R6}$ and ${\bf R9}$ (see Figure~\ref{mpt1}), the time elapsing in the timeout construct can be executed together with the delay-$\tell$ actions of $A_2$. After the timeout triggering, agent $A_1$ is however blocked, since $c_1$ is not entailed by the current empty store, and the precondition of the $\ask$ (rule ${\bf R3}$) is not satisfied. $A_2$ can execute the last delay-$\tell$, and then perform the $\tell(c_3)$ operation at $t=3$; the store becomes $\sigma = \bar{0} \otimes c_3 = c_3$. This finally unblocks $A_1$ at $t=4$, since, according to the precondition of rule ${\bf R3}$, $\sigma \sqsubseteq c_1$ (i.e., $c_3 \sqsubseteq c_1$). Finally, at $t=5$ we have $\langle  \success \parallel  \success, c_3 \rangle\textit{.}$

\begin{figure}[h]
\centering
\includegraphics[scale=1]{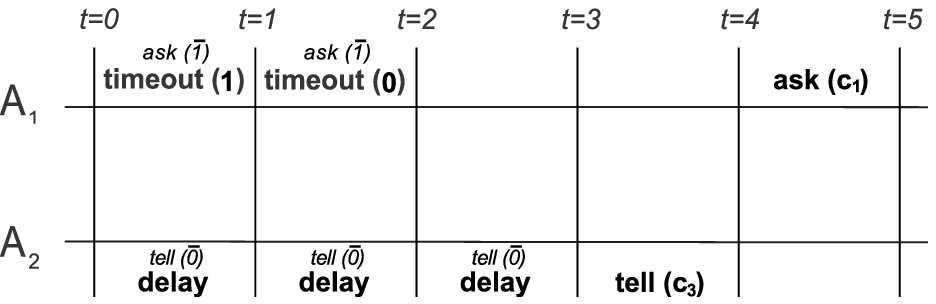} 
\caption{The  timeline of the execution of the $A_1 \parallel A_2$ agent  in Example~\ref{exmp2}.}
\label{figure:mptimeline2}
\end{figure}
\end{example}

\begin{example}[Watchdog]\label{exmp3} In this example
let $$A_1 :: \doo \; (\tell(c_1)\rightarrow\!^{+\infty}\, \ask(c_3)\rightarrow\!^{+\infty} \,\success)\; \watching^{+\infty} (c_2) \;\elsee $$ \vspace{-0.5cm} $$ \; (\,\tell(c_3)\rightarrow\!^{+\infty} \,\success\, )$$
and
$$ A_2::  \tell(c_2)\rightarrow^{+\infty} \success\textit{.}$$

We evaluate the following watchguard construct with two agents $A_1$ and $A_2$ in parallel:
$$ \langle( \doo \; (\tell(c_1)\rightarrow\!^{+\infty}\, \ask(c_3)\rightarrow\!^{+\infty} \,\success)\; \watching^{+\infty} (c_2) \;\elsee $$ \vspace{-0.5cm} $$ \; (\,\tell(c_3)\rightarrow\!^{+\infty} \,\success\, )  \parallel \tell(c_2)\rightarrow^{+\infty} \success), \bar{0} \rangle\textit{.}$$

According to Figure~\ref{fig:mpwatchex},  agent $A_1$ is translated in the following way, where the agent $B$ is a shorthand for the ``else'' branch of the watchdog, that is $\tell(c_3)\rightarrow^{+\infty} \success$:
$$\noow^{+\infty} \ c_2 \ \thene \ B \  \elsee \  (\tell(c_1)\rightarrow^{+\infty} \noow^{+\infty} \ c_2 \ \thene \  B \ \elsee\ \;$$ \vspace{-0.5cm} $$ (\ask(c_3)\rightarrow^{+\infty} \, \noow^{+\infty} \ c_2 \ \thene \ B \  \elsee \  \success)) \textit{.} $$

The execution timeline for this parallel agent is shown in Figure~\ref{figure:mptimeline3}. In the first time-unit we have that $\sigma= \bar{0} \not \sqsubseteq c_2$, i.e., the store does not imply the guard of the $\noow^{+\infty}$, and therefore the interruption of the watchguard in $A_1$ is not triggered yet. Thus, in the first time-unit, both $\tell(c_1)\rightarrow^{+\infty}$ of agent $A_1$ and $\tell(c_2)\rightarrow^{+\infty}$ of agent $A_2$ are executed. At time $t=1$, the interruption of the watchguard is immediately activated (i.e. $\noow^{+\infty} c_2$), since the store is now equal to $c_1 \otimes c_2= c_3$ and $c_3 \ent c_2$ (rule ${\bf R8}$ in Figure~\ref{mpt1}). Therefore, $\tell(c_3)\rightarrow^{+\infty}$ of agent $B$ in $A_1$ is executed, while $A_2$ already corresponds to the $\success$ agent).

\comment{Notice that the $\ask(c_3)\rightarrow^{+\infty}$ in agent $A_1$ can never be executed because the guard of the watchguard implies it, i.e. $c_2 \ent c_3$: the guard of $ \noow^{+\infty} \ c_2 \ \thene \
A \ \elsee \ B$ is always evaluated before $A$ or $B$, as explained by rule ${\bf R8}$  in Figure~\ref{mpt1}.}

\begin{figure}
\centering
\includegraphics[scale=0.84]{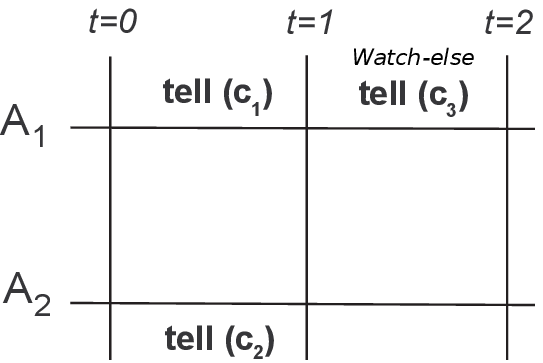} 
\caption{The  timeline of the execution of the $A_1 \parallel A_2$ parallel agent in Example~\ref{exmp3}.}
\label{figure:mptimeline3}
\end{figure}

\end{example}

\subsection{An Auction Example}\label{sec:mpauctionexample}
In Figure~\ref{mpexample} we model the negotiation and the management
of a generic service offered with a sort of auction: auctions, as
other forms of negotiation, naturally need both timed and
quantitative means to describe the interactions among
agents.  We reckon that an auction provides one of the most suitable example where to show the expressivity of the
{\itshape tsccp} language, since both time and preference (for a service or an object) are considered. In the following of the description we consider a  \emph{buyout} auction~\cite{buyout}, where the auctioneer improves the service and the related consumed resources (or, alternatively, its money price), bid after bid. When one (ore more) of the bidders agrees with the offer, it bids for it and the  auction is immediately declared as over.

The auctioneer (i.e. $AUCTIONEER$ in Figure~\ref{mpexample})
begins by offering a service described with the soft constraint
$c_{A_1}$. We suppose that the cost associated to the soft
constraint is expressed in terms of computational capabilities
needed to support the execution of the service: e.g., $c_{i} \sqsubseteq c_{j}$ means that
the service described by $c_{i}$ needs more computational resources
than $c_{j}$. By choosing the proper semiring, this load can be
expressed as a percentage of the CPU use, or in terms of money,
for example; we left this preference generic in the example, since we focus on the interaction among the agents.

We suppose that a constraint can be defined over
three domains of QoS features: availability, reliability and
execution time. For instance, $c_{A_1}$ is defined as $\mathit{availability}
> 95\% \wedge \mathit{reliability}
> 99\% \wedge \mathit{execution} \; \mathit{time} < 3\mathit{sec}$. Clearly, providing a
higher availability or reliability, and a lower execution time
implies raising the computational resources to support this improvement, thus worsening the
preference of the store.

\begin{figure}[h]  \scalebox{0.82}{
\centering\begin{tabular}{ l }
$\underline{AUCTIONEER}::$\\
$INIT\_A \longrightarrow $\\
$\tell(c_{A_1}) \stackrel{t_{sell}}{\longrightarrow}  \:(\Sigma_{i=1}^{n} \ask(bidder_i=i)\rightarrow^{a_A}\tell(winner=i)\rightarrow CHECK )\: \: \timeout({\it w_{A}})$ \\

$\;  (\tell(c_{A_2}) \stackrel{t_{sell}}{\longrightarrow}  \:(\Sigma_{i=1}^{n} \ask(bidder_i=i)\rightarrow^{a_A}  \tell(winner=i)\rightarrow CHECK )\: \: \timeout({\it w_{A}})$ \\

$\; \; ( \tell(c_{A_3}) \stackrel{t_{sell}}{\longrightarrow}  \:(\Sigma_{i=1}^{n} \ask(bidder_i=i)\rightarrow^{a_A}  \tell(winner=i)\rightarrow CHECK ) \: \: \timeout({\it w_{A}})\; $ \\

$\; \; \; \; \;\success ))$\\

\mbox{ }\\

$\underline{CHECK}::$ \\
$\doo \; ( \; ( \ask(service=end)\longrightarrow \success ) \; \; \timeout({\it w_{C}}) \; \; \tell({\it service=interrupt}) \;) $\\
$\; \; \; \; \; \; \; \; \watching_{\phi_\mathit{Check}} ({\it c_{check}}) \; \; \elsee  \; \; (\tell({\it service=interrupt}) \longrightarrow STOP_C)$ \\

\mbox{ }\\

$\underline{BIDDER_i}::$\\
$INIT\_B_i \longrightarrow $                                                      \\
$\doo \; (\; TASK_i \; )\; \watching_{\phi_\mathit{Bidder}} (c_{B_i}) \; \; \elsee  \; \;
\ask(\bar \1)\stackrel{t_{buy_i}}{\longrightarrow}
\tell(bidder_i=i) \longrightarrow$ \\
$\; \; \; \; \; \; \; \; (\;(\ask(winner = i) \longrightarrow USER_i) + (\ask(winner \neq i) \longrightarrow \success)\;)$ \\

\mbox{ }\\

$\underline{USER_i}::$\\
$\doo \; (\; USE\_SERVICE_i \longrightarrow \tell(service=end) \longrightarrow \success \; ) $\\
$\; \; \; \; \; \; \; \; {\watching_{\phi_\mathit{User}}}({\it service=interrupt}) \;
\;
{\elsee} \; \; ({\it STOP_i})$\\

\mbox{ }\\

$\underline{AUCTION\&MONITOR}:: \; AUCTIONEER \; \parallel \; BIDER_1
\; \parallel \; BIDDER_2 \; \parallel \; \dots \; \parallel BIDDER_n$

\end{tabular} }

\caption{An ``auction and management'' example for a generic
service}\label{mpexample}
\end{figure}

After the offer, the auctioneer gives time to the bidders (each of
them described with a possibly different agent $BIDDER_i$  in
Figure~\ref{mpexample}) to make their offer, since the choice of the
winner is delayed by $t_{sell}$ time-units (as in many real-world
auction schemes). A level $a_A$ is used to effectively check that
the global consistency of the store is enough good, i.e., the
computational power would not be already consumed under the given
threshold. After the winner is nondeterministically chosen among
all the bidders asking for the service, the auctioneer becomes a
supervisor of the used resource by executing the agent $CHECK$.
Otherwise, if no offer is received within $w_{A}$
time-units, a timeout interrupts the wait and the auctioneer
improves the offered service by adding a new constraint: for
example, in $\tell(c_{A_2})$, $c_{A_2}$ could be equivalent to
$execution \; time < 1sec$, thus reducing the latency of the
service (from $3$ to $1$ second) and consequently raising, at the
same time, its computational cost (i.e., $\sigma = c_{A_1} \otimes c_{A_2}
\sqsubseteq c_{A_1}$ means that we worsen the consistency level of the
store). The same offer/wait process is repeated three times in
Figure~\ref{mpexample}.

Each of the bidders in Figure~\ref{mpexample}
executes its own task (i.e., $TASK_i$, left generic since not in the scope of the example), but as soon as the offered
resource meets its demand (i.e. $c_{B_i}$
is satisfied by the store: $\sigma \sqsubseteq c_{B_i}$), the
bidder is interrupted and then asks to use the service. The time
needed to react and make an offer is modeled with $t_{buy_i}$:
fast bidders will have more chances to win the auction, if their
request arrives before the choice of the auctioneer. If one of the bidders
wins, then it becomes a user of the resource, by executing
$USER_i$. 

The agent  $USER_i$  uses the service (through the agent$USE\_SERVICE_i$, left generic in Figure~\ref{mpexample}), but it stops (using agent $STOP_i$, left generic in Figure~\ref{mpexample}) as soon as the service
is interrupted, i.e., as the store satisfies ${\it
service=interrupt}$. On the other side, agent $CHECK$  waits
for the use termination, but it interrupts the user if the
computation takes too long (more than $w_{C}$ time-units),
or if the user absorbs the computational capabilities beyond a
given threshold, i.e. as soon as the $c_{check}$ becomes implied
by the store (i.e. $\sigma \sqsubseteq c_{check}$): in fact,
$USE\_SERVICE_i$ could be allowed to ask for more power by
``telling'' some more constraints to the store. To interrupt the
service use, agent  $CHECK$  performs a $\tell({\it
service=interrupt})$. All the agents $INIT$, left generic in
Figure~\ref{mpexample}, can be used to initialize the computation.
In order to avoid a heavy notation in Figure~\ref{mpexample}, we do
not show the preference associated to constraints and the
consistency check label on the transition arrows, when they are
not significative for the example description. Also the $\phi_{\mathit{Check}}$, $\phi_{\mathit{Bidder}}$ and $\phi_{\mathit{User}}$ thresholds of the watchguard constructs are not detailed.

Finally, in the following we model a more refined behaviour of the auctioneer, which accepts the bidding with the highest value, where
$\underline{CHECK}$, $\underline{BIDDER_i}$ and $\underline{USER_i}$ are defined as in Figure \ref{mpexample}.

\begin{figure}[h]  \scalebox{0.82}{
\centering\begin{tabular}{ l }
$\underline{AUCTIONEER'}::$\\
$INIT\_A \longrightarrow $\\
$\tell(c_{A_1}) \stackrel{t_{sell}}{\longrightarrow}  \:(\Sigma_{i=1}^{n} \ask(bidder_i=i)\rightarrow^{a_A} CHOOSE)\: \: \timeout({\it w_{A}})$ \\

$\;  (\tell(c_{A_2}) \stackrel{t_{sell}}{\longrightarrow}  \:(\Sigma_{i=1}^{n} \ask(bidder_i=i)\rightarrow^{a_A}  CHOOSE)\: \: \timeout({\it w_{A}})$ \\

$\; \;  (\tell(c_{A_3}) \stackrel{t_{sell}}{\longrightarrow}  \:(\Sigma_{i=1}^{n} \ask(bidder_i=i)\rightarrow^{a_A}  CHOOSE ) \: \: \timeout({\it w_{A}})\;  \success ))$ \\

\mbox{ }\\

$\underline{CHOOSE}::$ \\
$
\begin{array}{ll}
  \noow \ (bidder_n=n) \ \thene \ \tell(winner=n)\rightarrow CHECK  \\
   \hspace*{0.75cm}\elsee \ (\ \noow \ (bidder_{n-1}=n-1) \ \thene \ \tell(winner=n-1)\rightarrow CHECK  \\
   \hspace*{1.5cm}\elsee \ (\dots (\ \noow \ (bidder_2=2) \ \thene \ \tell(winner=2)\rightarrow CHECK\\
   \hspace*{2.25cm}\elsee \ \tell(winner=1)\rightarrow CHECK) \dots))\text{.}
\end{array}
$\\

\mbox{ }\\

$\underline{newAUCTION\&MONITOR}:: \; AUCTIONEER' \; \parallel \; BIDER_1
\; \parallel \; BIDDER_2 \; \parallel \; \dots \; \parallel BIDDER_n$\\

\end{tabular} }

\caption{A new  ``auction and management'' example for a generic
service}\label{mpexample2}
\end{figure}

Many other real-life automated tasks can be modeled with the
{\itshape tsccp} language. For example,  a quality-driven
composition of web services: the agents that represent different
web services can add to the store their functionalities
(represented by soft constraints) with $\tell$ actions; the final
store models their composition. The consistency level of the store
 represents (for example) the total monetary cost of the
obtained service, or a value representing the consistency of the
integrated functionalities.
The reason is that, when we compose the
services offered by different providers, we cannot be sure of how
much they are compatible. A client wishing to use the
composed service can perform an $\ask$ with a threshold such that it
prevents the client from paying a high price, or having an unreliable
service. Softness is also useful  to model incomplete service
specifications that may evolve incrementally and, in general,
for non-functional aspects.


\section{The Denotational Model}\label{sec:mpdensem}
In this section we define a denotational characterization of the
operational semantics obtained by following the construction in
\cite{BGM00}, and by using {\em timed reactive sequences} to represent
{\itshape tsccp} computations. These sequences are similar to
those used in the semantics of dataflow languages~\cite{Jo85},
imperative languages~\cite{Br93} and  (timed) {\itshape
ccp}~\cite{BP90a,BGM00}.

The denotational model associates with a process a set of timed
reactive sequences of the form $\langle \sigma_1,\gamma_1\rangle
\cdots \langle \sigma_n,\gamma_n\rangle \langle \sigma
,\sigma\rangle$ where a pair of constraints $\langle
\sigma_i,\gamma_i\rangle$ represents a reaction of the given
process at time $i$: intuitively, the process transforms the global
store from $\sigma_i$ to $\gamma_i$ or, in other words, $\sigma_i$
is the assumption on the external environment while $\gamma_i$ is
the contribution of the process itself (which always entails  the
assumption). The last pair denotes a ``stuttering step'' in which
the agent ${\bf Success}$ has been reached. Since the basic
actions of {\itshape tsccp} are monotonic and we can also model a
new input of the external environment by a corresponding tell
operation, it is natural to assume that reactive sequences are
monotonic. Thus, in the following we  assume that each timed
reactive sequence $\langle \sigma_1,\gamma_1\rangle \cdots \langle
\sigma_{n-1},\gamma_{n-1}\rangle\langle \sigma_n,\sigma_n\rangle$
satisfies the conditions  $\gamma_i \vdash \sigma_i
\hbox{ and } \ \sigma_j \vdash \gamma_{j-1},  $ for any $i\in
[1,n-1]$ and $j\in [2,n]$.

The set of all reactive sequences is denoted by ${\cal S}$, its
typical elements by $s,s_1\ldots$, while sets of reactive
sequences are denoted by $S,S_1\ldots$, and $\varepsilon$ indicates
the empty reactive sequence. Furthermore, the symbol $\cdot$ denotes the
operator that concatenates sequences. In the following,
$Process$ denotes the set of {\itshape tsccp} processes.

Operationally, the reactive sequences of an agent are generated as follows.

\begin{definition}[Processes Semantics]\label{def:mpra}
We define the semantics ${\cal R}\in {\it Process}\rightarrow{\cal P}({\cal S})$ by
\[
\begin{array}[t]{ll}
{\cal R}({\it F\text{.}A})= & \{\la \sigma,\sigma'\ra \cdot w \in {\cal S} \mid
\la {\it A}, \sigma \ra \rrarrow \la {\it B}, \sigma' \ra
\hbox{ and } w\in {\cal R}(\it F\text{.} B) \}
\\
& \cup
\\
& \{\la \sigma,\sigma\ra \cdot w \in {\cal S} \mid \la {\it A},
\sigma \ra \not \rrarrow \hbox{ and }\\
& \hspace*{2.4cm}\begin{array}[t]{ll} \mbox{either }  {\it A}\neq {\bf Success} \hbox{ and } w\in {\cal R}(F\text{.}A) \\
\mbox{or }{\it A}={\bf Success} \hbox{ and } w\in {\cal R}(F\text{.} A) \cup
\{\varepsilon\} \}\textit{.}
\end{array}
\end{array}
\]

Formally  ${\cal R}$ is
defined as the least fixed-point of the operator $\Phi\in ({\it
Process}\rightarrow{\cal P}({\cal S})) \rightarrow {\it
Process}\rightarrow{\cal P}({\cal S})$ defined by

\begin{array}[t]{lll}
\Phi(I)({\it F\text{.} A})= & \{\la \sigma,\delta\ra \cdot w \in {\cal S}
\mid \la {\it A}, \sigma \ra \rrarrow \la {\it B}, \delta \ra
\hbox{ and } w\in I(\it F\text{.}  B) \}
\\
& \cup
\\
& \{ \la \sigma,\sigma\ra \cdot w \in {\cal S} \mid  \la {\it A},
\sigma \ra \not \rrarrow
\hbox{ and }\\
&\hspace*{2.4cm}\begin{array}[t]{ll}
\hbox{either }  {\it A}\neq {\bf Success} \hbox{ and } w\in I(F\text{.}A) \\
\mbox{or }{\it A}={\bf Success} \hbox{ and } w\in I(F\text{.} A) \cup
\{\varepsilon\}\}\textit{.}
\end{array}
\end{array}
\end{definition}
The ordering on ${\it Process}\rightarrow{\cal P}({\cal S})$ is
that of (point-wise extended) set-inclusion, and since it is
straightforward to check that $\Phi$ is continuous, standard
results ensure that the least fixpoint exists (and it is equal to
$\sqcup_{n\geq 0} \Phi^n(\bot)$).

Note that ${\cal R}(F\text{.}A)$ is the union of the set of all successful
reactive sequences that start with a reaction of $A$, and the set
of all successful reactive sequences that start with a stuttering
step of $A$.  In fact, when an agent is blocked, i.e., it cannot
react to the input of the environment, a stuttering step is
generated. After such a stuttering step, the computation can either
continue with the further evaluation of ${\it A}$ (possibly
generating more stuttering steps), or it can terminate if ${\it
A}$ is the ${\bf Success}$ agent. Note also that, since the ${\bf
Success}$ agent used in the transition system cannot make any
move, an arbitrary (finite) sequence of stuttering steps is always
appended to each reactive sequence.

\subsection{Correctness}\label{sec:mpcorrectness}

The observables ${\cal O}_{io}^{mp}(P)$ describing the input/output
pairs of successful computations can be obtained from ${\cal R}(P)$ by
considering suitable sequences, namely those sequences which do
not perform assumptions on the store. In fact, note that some
reactive sequences do not correspond to real computations:
Clearly, when considering a real computation no further
contribution from the environment is possible. This means that, at
each step, the assumption on the current store must be equal to
the store produced by the previous step. In other words, for any
two consecutive steps $\la \sigma_i,\sigma'_i\ra \la
\sigma_{i+1},\sigma'_{i+1}\ra$ we must have $\sigma'_i =
\sigma_{i+1}$. Thus, we are led to the following.

\begin{definition}[Connected  Sequences]
Let $s=\la \sigma_1,\sigma'_1\ra \la \sigma_2,\sigma'_2\ra \cdots
\la \sigma_{n},\sigma_{n}\ra$ be a reactive sequence. We say that
$s$ is connected if $\sigma_1= \1$ and $\sigma_i = \sigma'_{i-1}$
for each $i$, $2\leq i\leq n$.
\end{definition}

According to the previous definition, a sequence is connected if
all the information assumed on the store is produced by the
process itself. To be defined as
connected, a sequence must also have $\1$ as the initial
constraint. A connected sequence $s=\la \1,\sigma_1\ra \la \sigma_1,\sigma_2\ra \cdots
\la \sigma_{n},\sigma_{n}\ra$ represents a {\itshape tsccp}
computation of a process $F\text{.}A$,  where $\1$ is the input constraint and $\sigma_n \Downarrow _{Fv(A)}$ is the result.
From the above discussion we can derive the following property:

\begin{proposition}[Correctness] \label{theo:mpcorrectness}
For any process $P= F\text{.}A$ we have
\vspace{-0.2cm}
$${\cal O}_{io}^{mp}({ P}) = \{ \sigma_n \Downarrow _{Fv(A)}\mid
\begin{array}[t]{ll} \mbox{\rm there exists a connected sequence } s\in
{\cal R}(P) \ \mbox{\rm such that }
\\
 s=\la \1,\sigma_1\ra \la \sigma_1, \sigma_2\ra\cdots
 \la \sigma_{n},\sigma_{n}\ra \}\textit{.}
\end{array}
$$
\end{proposition}
\begin{proof}
From the close correspondence between the
rules of the transition system and the definition of the
denotational semantics, we have that $s \in {\cal R}(P)$ if and only if
$s=\la \sigma_1,\sigma_1'\ra \la \sigma_2,\sigma_2'\ra\cdots
  \la \sigma_{n},\sigma_{n}\ra$, $A_1=A$, $A_n ={\bf Success}$  and for $i\in [1,n-1]$,
\begin{itemize}
  \item either $\la {A_i},\sigma_i\ra \rrarrow \la {A_{i+1}},\sigma'_i\ra$
  \item or $\la {A_i},\sigma_i\ra \not\rrarrow$, $A_{i+1}= A_i$  and $\sigma'_i= \sigma_i$.
\end{itemize}
Then there exists a connected sequence $s\in {\cal R}(P)$ if and only if
$s=\la \sigma_1,\sigma_2\ra \la \sigma_2,\sigma_3\ra$ $\cdots
  \la \sigma_{n},\sigma_{n}\ra$, $A_1=A$, $\sigma_1= \1$, $A_n ={\bf Success}$  and for $i\in [1,n-1]$,
$\la {A_i},\sigma_i\ra \rrarrow \la {A_{i+1}},\sigma_{i+1}\ra$.
Therefore, the proof follows by definition of ${\cal O}_{io}^{mp}(P)$.
\end{proof}
\medskip

\subsection{Compositionality of the Denotational Semantics for {\itshape tsccp}
Processes}\label{sec:mpcompdensem}

In order to prove the compositionality of the denotational
semantics, we now introduce a semantics $\os F\text{.}A \cs(e)$, which is
compositional by definition and where, for technical reasons, we
explicitly represent  the environment $e$ that associates a
denotation to each procedure identifier. More precisely, assuming
that {\itshape Pvar} denotes the set of procedure identifiers,
${\it Env}={\it Pvar}\rightarrow {\cal P}({\cal S})$, with typical
element $e$, is the set of {\itshape environments}. Given $e \in
Env$, $p \in {\it Pvar}$  and $f \in {\cal P}({\cal S})$, we
denote by $e'=e\{f/p\}$ the new environment such that $e'(p)=f$
and $e'(p')= e(p')$ for each procedure identifier $p'\neq p$.

Given a process $F\text{.}A$, the denotational semantics $\os F\text{.}A \cs
:{\it Env} \rightarrow {\cal P}({\cal S})$ is defined by the
equations in Figure~\ref{mpdensem}, where $\mu$ denotes the least
fixpoint with respect to the subset inclusion of elements of ${\cal P}({\cal
S})$. The semantic operators appearing in Figure~\ref{mpdensem} are
formally defined as follows; intuitively they reflect the operational behavior of their syntactic
counterparts in terms of
reactive sequences.\footnote{In Figure~\ref{mpdensem} the syntactic operator
$\rightarrow_i$ is either of
 the form $\rightarrow^{a_i}$ or $\rightarrow_{\phi_i}$.} We first need the following definition.

\begin{definition}\label{def:servizio}
Let $\sigma, \phi$  and
$c$ be constraints in ${\cal C}$ and let $a \in \mathcal{A}$. We
say that

\begin{itemize}
    \item $\sigma   \succ ^{a}\, c, $ if $ (\sigma\ent c $ and $\sigma\Downarrow_{\emptyset} \not< a) \ \ \  \hbox {while } \ \ \  \sigma \succ _{\phi}\,   c, $ if $ (\sigma \ent c$ and $\sigma \not\sqsubset
    \phi)$.
\end{itemize}
\end{definition}

\begin{definition}[Semantic operators]\label{def:mpsemoperators}

Let $S,S_i$ be sets of reactive sequences,
 $c,c_i$ be constraints and let $\succ_i$ be either of
 the form $\succ^{a_i}$ or $\succ_{\phi_i}$.
Then we define the operators $\tilde{tell}$, $\tilde{\sum}$,
$\tilde{\parallel}$, $\tilde{now}$ and $\tilde{\exists} x$ as
follows:

\noindent {\bf The (valued) tell operator}
\[\begin{array}{ll}
  \tilde{tell}^a(c,  S)= &
\begin{array}[t]{l}
\{ s  \in {\cal S} \mid s = \la \sigma ,\sigma\otimes c\ra \cdot
s', \ \sigma\otimes c\Downarrow_{\emptyset} \not<a  \mbox{ and
} s' \in S\ \}\textit{.}
\end{array} \\
\end{array}
\]
\[\begin{array}{ll}
  \tilde{tell}_{\phi}(c,  S)= &
\begin{array}[t]{l}
\{ s  \in {\cal S} \mid s = \la \sigma ,\sigma\otimes c\ra \cdot
s', \ \sigma\otimes c \not\sqsubset
    \phi  \mbox{
and } s' \in S\ \}\textit{.}
\end{array} \\
\end{array}
\]

\noindent
{\bf The guarded choice}
\[
{\tilde{\sum}} _{i=1}^n c_i \succ_i \, S_i =
\begin{array}[t]{l}
\{ s\cdot s'   \in {\cal S} \mid
\begin{array}[t]{l}
s = \la \sigma_1 ,\sigma_1\ra \cdots \la \sigma_m,\sigma_m\ra,
\sigma_j {\not\succ_i }\,c_i
\\
\mbox{for each } j\in [1,m$-$1], i\in [1,n],
\\
\sigma_m \succ_h\,  c_h \hbox { and } s'\in S_h \hbox{
for an } h\in[1,n] \ \}\textit{.}
\end{array}
\end{array}
\]

\noindent {\bf The parallel composition} Let $\tilde{\parallel}
\in{\cal S}\times {\cal S}\rightarrow {\cal S}$ be the
(commutative and associative) partial operator defined as follows:
\[
\begin{array}{ll}
\la \sigma_1,\sigma_1 \otimes \gamma_1\ra \cdots\la
\sigma_n,\sigma_n \otimes \gamma_n\ra \la \sigma,\sigma\ra \
\tilde{\parallel}\ \la \sigma_1,\sigma_1 \otimes
 \delta_1\ra\cdots\la \sigma_n,\sigma_n \otimes \delta_n\ra \la \sigma,\sigma\ra & =
\\
\la \sigma_1, \sigma_1 \otimes \gamma_1\otimes \delta_1\ra\cdots
\la \sigma_n, \sigma_n \otimes \gamma_n\otimes \delta_n \ra\la
\sigma,\sigma\ra\textit{.}
\end{array}
\]

We define $S_1\tilde{\parallel} S_2$ as the point-wise extension
of the above operator to sets.

\noindent
{\bf The (valued) now operator}
\[\begin{array}{ll}
  \tilde{now}^a(c,  S_{1} , S_{2})=
  \{ s  \in {\cal S} \mid & s = \la \sigma,\sigma' \ra \cdot s', \, \sigma\Downarrow_{\emptyset} \not<a \hbox{ and }\\
  &\hbox{either }
  \sigma\ent c  \mbox{ and } s \in S_{1}\\
  & \hbox{or } \sigma\not \ent c \mbox{ and } s \in S_{2}\ \}\textit{.} \\
\end{array}
\]
\[\begin{array}{ll}
  \tilde{now}_{\phi} (c,  S_{1} , S_{2})=
  \{ s  \in {\cal S} \mid & s = \la \sigma ,\sigma' \ra \cdot s', \sigma \not\sqsubset \phi \hbox{ and }\\
  &\hbox{either }
 \sigma\ent c  \mbox{ and } s \in S_{1}\\
  & \hbox{or }  \sigma\not \ent c  \mbox{ and } s \in S_{2}\ \}\textit{.} \\
\end{array}
\]

\noindent
{\bf The hiding operator}
The semantic hiding operator can be defined as
follows:
\[
\begin{array}{ll}
  {\bf \tilde{ \exists} } x S =  \{ s \in {\cal S} \mid & \mbox{\rm
  there exists $ s' \in S$ such that $s
= s'[x/y]$ with $y$ new } \}
\end{array}
\]
where $s'[x/y]$ denotes the sequence obtained from $s'$ by
replacing the variable $x$ for the variable $y$, which we assume to
be new.\footnote{To be more precise, we assume that each time that we
consider a new application of the operator  ${\bf \tilde{
\exists} }$ we use a new, different $y$. As in the case of the
operational semantics, this can be ensured by a suitable renaming
mechanism.}
\end{definition}

Obviously, the semantic (valued) tell operator reflects  the operational behavior of the syntactic (valued) tell.  Concerning the semantic
choice operator, a sequence in $\tilde{\sum} _{i=1}^n c_i
\succ_i \,S_i$ consists of an initial period of waiting
for a store which satisfies one of the guards. During this waiting period, only the environment is active by
producing the constraints $\sigma_j$, while the process itself
generates the stuttering steps $\langle \sigma_j,\sigma_j\rangle$.
When the store is strong enough to satisfy a guard, that is to entail a
$c_h$ and to satisfy the condition on the cut level,
the resulting sequence is  obtained by adding $s'\in S_h$ to
the initial waiting period.
In the semantic parallel operator defined on sequences, we require
that the two arguments of the operator agree at each point of time
with respect to the contribution of the environment (the
$\sigma_i$'s), and that they have the same length (in all other
cases the parallel composition is assumed being undefined).

If $F\text{.}A$ is a closed process, that is if all the procedure names
occurring in $A$ are defined in $F$, then  $\os F\text{.}A\cs(e)$ does
not depend on $e$, and it will be indicated as $\os F\text{.}A\cs$.
Environments in general allow us to define the semantics also of
processes that are not closed. The following result shows the correspondence between the two
semantics we have introduced and, therefore, it proves the compositionality of
${\cal R}(F\text{.}A)$. From the above discussion we can derive the following property:

\begin{proposition}[Compositionality]\label{theo:mpreqden}
If $F\text{.}A$ is closed then ${\cal R}(F\text{.}A) = \os F\text{.}A \cs$ holds.
\end{proposition}
\begin{proof} We prove by induction on the complexity of the agent $A$ that
\[
\begin{array}{ll}
  \os F\text{.}A \cs=\{ s  \mid & s=\la \sigma_1,\sigma_1'\ra \la \sigma_2,\sigma_2'\ra\cdots
  \la \sigma_{n},\sigma_{n}\ra, \\
  & A_1=A, A_n ={\bf Success} \mbox{ and for } i\in [1,n-1],\\
  & \mbox{either } \la {A_i},\sigma_i\ra \rrarrow \la
  {A_{i+1}},\sigma'_i\ra \\
  & \mbox{or } \la {A_i},\sigma_i\ra \not \rrarrow , \,
  A_{i+1}=A_i,\, \sigma'_i= \sigma_i    \}\textit{.}
\end{array} \]
Then the proof follows by definition of ${\cal R}(P)$.

When the $P$ is not of the form $F\text{.}B \parallel C$ the thesis
follows immediately from the close correspondence between the
rules of the transition system and the definition of the
denotational semantics.

Assume now that $P$ is of the form $F\text{.}B \parallel C$.
By definition of the denotational semantics, $s \in \os F\text{.}A \cs$ if and only if
$s=\la \sigma_1,\sigma_1'\ra \la \sigma_2,\sigma_2'\ra\cdots
  \la \sigma_{n},\sigma_{n}\ra$ and there exist $s' \in
\os F\text{.}B\cs $ and $s'' \in \os F\text{.}C\cs $,

\[
\begin{array}{ll}
s'= \la \sigma_1,\sigma_1 \otimes \gamma_1\ra  \la \sigma_2,\sigma_2 \otimes \gamma_2\ra \cdots\la
\sigma_n,\sigma_n \ra \\
s''= \la \sigma_1,\sigma_1 \otimes
 \delta_1\ra\la \sigma_2,\sigma_2 \otimes \delta_2\ra\cdots\la \sigma_n,\sigma_n\ra
\end{array}
\]
such
that  for each $i\in [1,n-1]$,
$\sigma_i' = \sigma_i \otimes \gamma_i \otimes \delta_i$.
By inductive hypothesis $s' \in
\os F\text{.}B\cs $ and $s'' \in \os F\text{.}C\cs $ if and only if for  $i\in [1,n-1],$
\begin{eqnarray}
\begin{array}{ll}
\mbox{either  } \la {B_i},\sigma_i\ra \rrarrow \la {
B_{i+1}},\sigma_i \otimes \gamma_i \ra ,\\
\mbox{ or  } \la {B_i},\sigma_i\ra \not \rrarrow , \,
B_{i+1}=B_i,\,  \sigma_i \otimes \gamma_i = \sigma_i  & \mbox{and} \\
\mbox{either  } \la {C_i},\sigma_i\ra \rrarrow \la {
{C_{i+1}}},\sigma_i \otimes \delta_i\ra, \\
 \mbox{ or  } \la {C_i},\sigma_i\ra \not \rrarrow , \,
C_{i+1}=C_i,\,  \sigma_i \otimes \delta_i = \sigma_i  \text{.} \\
\end{array}\label{24feb123}
\end{eqnarray}
$B_1=B$,  $B_n={\bf Success}$, $C_1=C$  and $C_n={\bf Success}$.
Therefore, by Rule {\bf R8} and previous observations, we have that
(\ref{24feb123}) holds if and only if   $B_1\parallel C_1 =B\parallel C$, $B_n\parallel C_n ={\bf Success}$  and for $i\in [1,n-1]$,
\[\begin{array}{ll}
    \mbox{either } \la B_i \parallel C_i,\sigma_i\ra \rrarrow \la B_{i+1} \parallel
C_{i+1},\sigma_i'\ra & \\
\mbox{or } \la B_i \parallel C_i,\sigma_i\ra \not \rrarrow, \,
  A_{i+1}\parallel  B_{i+1}=A_i\parallel  B_{i},\, \sigma'_i= \sigma_i
  \end{array}
\]
and then the thesis.
\end{proof}
\medskip

\begin{figure}
\centering
\begin{tabular}{ l  l  l}

\mbox{ }&\mbox{ }&\\

\mbox{{\bf E1}} & $\os {\it F\text{.}{\bf success}} \cs(e) = \{ \la
\sigma_1, \sigma_1\ra \la \sigma_2, \sigma_2\ra  \cdots \la
\sigma_n, \sigma_n\ra \in {\cal S}
\mid n\geq 1\}$&\\

\mbox{ }&\mbox{ }&\\

\mbox{{\bf E2}} & $\os {\it F\text{.}\hbox{\tell}(c) \rightarrow^{a} A}
\cs(e) =  \tilde{tell}^a (c,  \os {\it F\text{.}A}\cs(e))$&\\

\mbox{ }&\mbox{ }&\\

\mbox{{\bf E3}} & $\os {\it F\text{.}\hbox{\tell}(c) \rightarrow_{\phi}
A}
\cs(e) =  \tilde{tell}_{\phi}(c,  \os {\it F\text{.}A} \cs(e))$&\\

\mbox{ }&\mbox{ }&\\

\mbox{{\bf E4}} & $\os {\it F\text{.}\sum_{i=1}^n }{\it \ask(c_i)
\rightarrow_i  A_i} \cs(e) =
{\tilde{\sum} } _{i=1}^n c_i \succ_i \,\os {\it F\text{.}A_i} \cs(e)$&\\

\mbox{ }&\mbox{ }&\\

\mbox{{\bf E5}} & $\os {\it F\text{.}\noow^a\ c\  \thene \ A
\ \elsee \ B}\cs(e)=
\tilde{{\it now}}^a(c,\os{\it F\text{.}A}\cs(e),\os {\it F\text{.}B}\cs(e))$&\\

\mbox{ }&\mbox{ }&\\

\mbox{{\bf E6}} & $\os {\it F\text{.}\noow_\phi\ c\  \thene \ A \
\elsee \ B}\cs(e)=
\tilde{{\it now}}_\phi(c,\os{\it F\text{.}A}\cs(e),\os {\it F\text{.}B}\cs(e))$&\\

\mbox{ }&\mbox{ }&\\

\mbox{{\bf E7}} & $\os {\it F\text{.}A\parallel B}\cs(e) = \os {\it
F\text{.}A}\cs(e) \ \tilde{ \parallel }
\ \os {\it G\text{.}B}\cs(e)$&\\

\mbox{ }&\mbox{ }&\\

\mbox{{\bf E8}} &
$\os {\it F\text{.}\exists x A}\cs(e) = \tilde{\exists}x \os{\it F\text{.}A} \cs(e)$&\\

\mbox{ }&\mbox{ }&\\

\mbox{{\bf E9}} &
$\os {\it F\text{.}p(x)} \cs(e) = \mu \Psi \; \; \text{where}  \; \;
 \Psi(f)=\os F\setminus\{p\}\text{.} ask(\bar{\1}) \rightarrow
A\cs(e\{f/p\}),  \; \; p(x)::A \in F$\\

\mbox{ }&\mbox{ }&\\
\end{tabular}
\caption{The semantics  $\os F\text{.}A \cs$(e).}\label{mpdensem}
\end{figure}

\section{An Interleaving Approach for non-Time-elapsing Actions}\label{sec:interleaving}
In this section, we show a different version of the {\itshape tsccp} language:
while in  {\itshape tsccp} the parallel operator is modeled in terms of {\itshape maximal parallelism}, the same operator can be treated also in terms of interleaving. According to {\itshape maximal parallelism}, at each
moment every enabled agent of the system is activated, while in
the second paradigm an agent could not be assigned to a ``free''
processor. Clearly, since we have dynamic process creation, a maximal
parallelism approach has the disadvantage that, in general, it
implies the existence of an unbound number of processes. On the
other hand a naive interleaving semantic could be problematic from
the time viewpoint, as in principle the time does not pass for
enabled agent which are not scheduled. For the semantics in this section we follow a solution
analogous to that one adopted in~\cite{BoGaMe04}: we
assume that the parallel operator is interpreted in terms of
interleaving, as usual, however we must assume maximal parallelism
for actions depending on time. In other words, time passes for all
the parallel processes involved in a computation. To summarize, in this section we
adopt maximal parallelism for time elapsing (i.e. for
timeout constructs) and an interleaving model
for basic computation steps (i.e. (valued) \ask \  and (valued)
\tell \ actions).

To distinguish this new approach, we named the resulting language as
{\itshape tsccp-i}, i.e., {\itshape tsccp} with interleaving. Time-outs are modeled in {\itshape tsccp-i} by the construct $\hbox{\bf askp}_t(c)?_{\phi} A \mathit{:} B$ which replaces the $\noow _{\phi}\ c \ \thene \ A \ \elsee \ B$ construct of {\itshape tsccp} and directly has time $t$ as one of its parameters, differently from the $\noow_{\phi}$ agent. The $\hbox{\bf askp}_t$ agent can be interpreted as follows: one is allowed to
wait $t$ time-units for the entailment of the constraint $c$ by the
store and the subsequent evaluation of the process $A$; if
this time limit is exceeded, then the process $B$ is
evaluated.  Analogously for the construct $
 \hbox{\bf askp}_t(c)?^a A \mathit{:} B$.

\begin{definition}[{\itshape tsccp-i}]\label{def:tsccilanguage}
Given a soft constraint system $\langle S,D,V\rangle$, the
corresponding structure $\C$, any semiring value $a$,  {\em soft constraints} $\phi, c \in {\cal C}$ and any tuple of variables $x$,
the syntax of the {\itshape tsccp-i}
language is given by the following grammar:
\[
\begin{array}{ll}
P ::= & F \text{.} A\\
F ::= & p(x):: A \;|\; F.F\\
A ::= & {\bf success} \;|\; \hbox{\tell}(c)\rarrow_\phi A
\;|\;\hbox{\tell}(c)\rarrow ^a A \;|\; E \;|\; A\parallel A
\;|\;
 \exists x A \;|\; p(x) \;|\;  \\
 & \Sigma_{i=1}^{n} E_i\;|\; \hbox{\bf askp}_t(c)?_{\phi} A  \mathit{:} A \;|\;
 \hbox{\bf askp}_t(c)?^a A  \mathit{:} A \\
E ::= &\hbox{\bf ask}(c) \rarrow_\phi A  \;|\; \hbox{\bf ask}(c)
\rarrow^a A \\
\end{array}
\]
where, as in Definition \ref{def:tscclanguage}, $P$ is the class of processes, $F$ is the class
of sequences of procedure declarations (or clauses), $A$ is the
class of agents. As before, in a {\itshape tsccp-i} {\em process} $P=F\text{.}A$,
$A$ is the initial agent, to be executed in the context of the set of declarations $F$.
\end{definition}

Analogously to {\itshape tsccp}
processes, in order to simplify the notation, in the following
we will usually write a  {\itshape tsccp-i} {\em process} $P=F\text{.}A$
 simply as the corresponding agent $A$.

The operational model of {\itshape tsccp-i} processes  can be formally described
by a labeled transition system $T= ({\it Conf}, Label, \Rrarrow
)$, where we assume that each transition step exactly takes  one
time-unit.  Configurations (in) {\it Conf} are pairs consisting of a
process  and a constraint in ${\cal C}$
representing the common {\em store}.
$\mathcal{L} = \{\tau,\omega\}$ is the set of labels. We use
labels to distinguish ``real'' computational steps performed by
processes which have the control (label $\omega$) from the
transitions which model only the passing of time (label $\tau$).
So $\omega$-actions are those performed by processes that modify
the store ($\bf{tell}$),  perform a check on the store
(${\bf ask}$, ${\bf askp}_t$),  correspond to exceeding a time-out
(${\bf askp}_0$), or perform a choice ($\Sigma_{i=1}^{n} E_i$). On the
other hand, $\tau$-actions are those performed by time-out
processes (${\bf askp}_t$) in case they have not the control.
In Figure~\ref{t2} we show the semantics of all the {\itshape tsccp-i} actions,
but in the following we describe only the actions whose semantics is different
from that one presented in Figure~\ref{mpt1} (i.e., for {\itshape tsccp}), that is we describe in detail the parallelism and
the ${\bf askp}_t$ agent. The semantics of the other actions of {\itshape tsccp-i}
is the same as for {\itshape tsccp}, except for the fact that their transition is labeled with $\omega$.

\begin{figure}
    \begin{center}
    \scalebox{0.9}{
\begin{tabular}{llll}
&\mbox{   }&\mbox{   } &\mbox{   }
\\
\mbox{\bf Q1}& $\frac {\displaystyle (\sigma \otimes
c)\Downarrow_{\emptyset} \not< a}{\displaystyle
\begin{array}{l}
\la \hbox{\tell}(c) \rightarrow^{a} A, \sigma \ra \stackrel{\omega}{\Rrarrow} \la A,
\sigma \otimes c\ra
\end{array}}$\ \ \ & \bf{V-tell}&
\\
&\mbox{   }&\mbox{   } &\mbox{   }
\\
\mbox{\bf Q2}& $\frac {\displaystyle \sigma \otimes c
\not\sqsubset \phi}{\displaystyle
\begin{array}{l}
\la \hbox{\tell}(c)\rightarrow_{\phi} A, \sigma \ra \stackrel{\omega}{\Rrarrow} \la
A, \sigma \otimes c\ra
\end{array}}$ & \bf{Tell} &
\\
&\mbox{   }&\mbox{   } &\mbox{   }
\\
\mbox{\bf Q3}& $\frac {\displaystyle \sigma \ent c \ \ \ \ \
\sigma\Downarrow_{\emptyset} \not< a}{\displaystyle
\begin{array}{l}
\la \hbox{\ask}(c) \rightarrow^{a} A, \sigma \ra \stackrel{\omega}{\Rrarrow} \la A,
\sigma \ra
\end{array}}$\ \ \ & \bf{V-ask}&
\\
&\mbox{   }&\mbox{   } &\mbox{   }
\\
\mbox{\bf Q4}& $\frac {\displaystyle \sigma \ent c \ \ \ \  \sigma
\not\sqsubset \phi}{\displaystyle
\begin{array}{l}
\la \hbox{\ask}(c) \rightarrow_{\phi} A, \sigma  \ra \stackrel{\omega}{\Rrarrow} \la
A, \sigma \ra
\end{array}}$ & \bf{Ask}&
\\
&\mbox{   }&\mbox{   } &\mbox{   }
\\

\mbox{\bf Q5 }& $\frac {\displaystyle \la A,\sigma \ra \stackrel{\xi}{\Rrarrow}
\la A', \sigma' \ra\ \ \ \ \la B,\sigma\ra \stackrel{\tau}{\Rrarrow}
\la B', \sigma \ra \ \ \ \ \xi\in\{\tau,\omega\}}
{\displaystyle
\begin{array}{l}
\la A\parallel B,\sigma\ra \stackrel{\xi}{\Rrarrow}\la A'\parallel B',\sigma' \ra
\\
\la B\parallel A,\sigma\ra \stackrel{\xi}{\Rrarrow}\la B'\parallel A',\sigma' \ra
\end{array}}$ & \bf{Parall1}&
\\
&\mbox{   }&\mbox{   }&\mbox{   }
\\
\mbox{\bf Q6 }&
$\frac {\displaystyle \la A,\sigma \ra \stackrel{\xi}{\Rrarrow}
\la A', \sigma' \ra\ \ \ \ \la B,\sigma\ra \not\stackrel{\tau}{\Rrarrow}
 \ \ \ \ \xi\in\{\tau,\omega\}}
{\displaystyle
\begin{array}{l}
\la A\parallel B, \sigma \ra\stackrel{\xi}{\Rrarrow} \la A'\parallel B, \sigma'
\ra
\\
\la B\parallel A, \sigma  \ra\stackrel{\xi}{\Rrarrow} \la B\parallel A', \sigma'
\ra
\end{array}}$&
 \bf{Parall2}&
\\
&\mbox{   }&\mbox{   }&
\\

\mbox{\bf Q7}& $\frac {\displaystyle \la  E_j,\sigma\ra \stackrel{\omega}{\Rrarrow}
\la A_j,\sigma' \ra\ \ \ \ \ \ j\in [1,n]} {\displaystyle
\begin{array}{l}
\la \Sigma_{i=1}^{n}E_i , \sigma \ra \stackrel{\omega}{\Rrarrow}\la A_j,\sigma'\ra
\end{array}}$ & \bf{Nondet}&
\\
&\mbox{   }&\mbox{   }&
\\
\mbox{\bf Q8 }
& $\la p(x),\sigma\ra\stackrel{\omega}{\Rrarrow}\la A,
\sigma\ra \ \ \ \ {\it p(x) :: A \in F}$
&\bf{P-call}&
\\
&\mbox{   }&\mbox{   }&
\\
\mbox{\bf Q9 }& $\frac {\displaystyle \la A[x/y],  \sigma
\ra\stackrel{\xi}{\Rrarrow}\la B, \sigma' \ra \ \ \ \ \xi\in\{\tau,\omega\}} {\displaystyle\la \exists x A,
\sigma \ra\stackrel{\xi}{\Rrarrow}\la B, \sigma' \ra}$
&\bf{Hide} &

\\

&\mbox{   }&\mbox{   } &\mbox{   }
\\
\mbox{\bf Q10 } & $\frac {\displaystyle \sigma \ent c \ \ \ \
\sigma\Downarrow_{\emptyset} \not< a \ \ \ \  t>0} {\displaystyle
\begin{array}{l}
\la \hbox{\bf askp}_t(c)?^{a} A  \mathit{:} B,\sigma\ra\stackrel{\omega}{\Rrarrow}\la A,  \sigma \ra
\end{array}}$  & \bf{V-askp1}&
\\
&\mbox{   }&\mbox{   }&
\\
\mbox{\bf Q11 }
&$\frac {\displaystyle
\sigma\Downarrow_{\emptyset} < a\ \ \ \  t>0} {\displaystyle
\begin{array}{l}
\la \hbox{\bf askp}_t(c)?^a A  \mathit{:} B,\sigma\ra\stackrel{\omega}{\Rrarrow}\la B, \sigma\ra
\end{array}}$ & \bf{V-askp2}&
\\
&\mbox{   }&\mbox{   }&
\\
\mbox{\bf Q12 }
&$\frac {\displaystyle  \sigma \not \ent c
\ \ \ \
\sigma\Downarrow_{\emptyset} \not< a \ \ \ \  t>0} {\displaystyle
\begin{array}{l}
\la \hbox{\bf askp}_t(c)?^a A  \mathit{:} B,\sigma\ra\stackrel{\omega}{\Rrarrow}\la\hbox{\bf askp}_{t-1}(c)?^a A  \mathit{:} B,\sigma\ra
\end{array}}$ & \bf{V-askp3}&
\\
&\mbox{   }&\mbox{   }&
\\
\mbox{\bf Q13 }&    $
{\displaystyle
\begin{array}{l}
\la \hbox{\bf askp}_t(c)?^a A  \mathit{:} B,\sigma\ra\stackrel{\tau}{\Rrarrow}
\la \hbox{\bf askp}_{t-1}(c)?^a A  \mathit{:} B,\sigma\ra \ \ \ \  t>0
\end{array}}$& \bf{V-askp4}&
\\
&\mbox{   }&\mbox{   }&
\\
\mbox{\bf Q14}&$ {\displaystyle
\begin{array}{l}
\la \hbox{\bf askp}_0(c)?^a A  \mathit{:} B,\sigma\ra\stackrel{\omega}{\Rrarrow} \la B,  \sigma\ra
\end{array}}$& \bf{V-askp5}&
\\
&\mbox{   }&\mbox{   }&
\\
\mbox{\bf Q15}
& $\frac {\displaystyle \sigma \ent c \ \ \ \ \sigma
\not\sqsubset \phi \ \ \ \  t>0} {\displaystyle
\begin{array}{l}
\la \hbox{\bf askp}_t(c)?_{\phi} A  \mathit{:} B,\sigma\ra\stackrel{\omega}{\Rrarrow}\la A,  \sigma \ra
\end{array}}$  & \bf{Askp1}&
\\
&\mbox{   }&\mbox{   }&
\\
\mbox{\bf Q16 }
& $\frac {\displaystyle  \sigma
\sqsubset \phi\ \ \ \  t>0} {\displaystyle
\begin{array}{l}
\la \hbox{\bf askp}_t(c)?_{\phi} A  \mathit{:} B,\sigma\ra\stackrel{\omega}{\Rrarrow}\la B, \sigma\ra
\end{array}}$  & \bf{Askp2}&
\\
&\mbox{   }&\mbox{   }&
\\
\mbox{\bf Q17}
& $\frac {\displaystyle  \sigma \not \ent c \ \ \ \ \sigma
\not \sqsubset \phi\ \ \ \  t>0} {\displaystyle
\begin{array}{l}
\la \hbox{\bf askp}_t(c)?_{\phi} A  \mathit{:} B,\sigma\ra\stackrel{\omega}{\Rrarrow}\la \hbox{\bf askp}_{t-1}(c)?_{\phi} A  \mathit{:} B,\sigma\ra
\end{array}}$  & \bf{Askp3}&
\\
&\mbox{   }&\mbox{   }&
\\
\mbox{\bf Q18}
&    $
{\displaystyle
\begin{array}{l}
\la \hbox{\bf askp}_t(c)?_{\phi} A  \mathit{:} B,\sigma\ra\stackrel{\tau}{\Rrarrow}
\la \hbox{\bf askp}_{t-1}(c)?_{\phi} A  \mathit{:} B,\sigma\ra \ \ \ \  t>0
\end{array}}$         & \bf{Askp4}&
\\
&\mbox{   }&\mbox{   }&
\\
\mbox{\bf Q19}
& $ {\displaystyle
\begin{array}{l}
\la \hbox{\bf askp}_0(c)?_{\phi} A  \mathit{:} B,\sigma\ra\stackrel{\omega}{\Rrarrow} \la B,  \sigma\ra
\end{array}}$& \bf{Askp5}&\\
&\mbox{   }&\mbox{   }&
\\
\end{tabular} }
    \end{center}

  \caption{The transition system for {\itshape tsccp-i}.}\label{t2}
\end{figure}

\begin{description}
\item[Parallelism] Rules {\bf Q5} and {\bf Q6} in Figure~\ref{t2} model the parallel
composition operator  in
terms of {\em interleaving}, since only one basic $\omega$-action
is allowed for each transition (i.e. for each unit of time). This
means that the access to the shared store is granted to one process a time. However, time
passes for all the processes appearing in the $\parallel$ context
at the external level, as shown by rule {\bf Q5}, since
$\tau$-actions are allowed together with a $\omega$-action. On the
other hand, a parallel component is allowed to proceed in
isolation if (and only if) the other parallel component cannot
perform a $\tau$-action (rule {\bf Q6}). To summarize, we adopt
maximal parallelism for time elapsing (i.e. $\tau$-actions) and an
interleaving model for basic computation steps (i.e.
$\omega$-actions).

We have adopted this approach because it seems more adequate to the nature of
time-out operators not to interrupt the elapsing of time, once the
evaluation of a time-out has started. Clearly one could start the
elapsing of time when the time out process is scheduled, rather
than when it appears in the top-level current parallel context.
This modification could easily be obtained by adding a syntactic
construct to differentiate active timeouts from inactive ones,
and by accordingly changing  the transition system. One could also
easily modify the semantics (both operational and denotational) to
consider a more liberal assumption which allows multiple ask
actions in parallel.

\item[Valued-Askp$_t$] The rules ${\bf Q10}$-${\bf Q14}$ in Figure~\ref{t2}
show that the time-out process $\hbox{\bf askp}_t(c)?^a A  \mathit{:} B$
behaves as $A$ if $c$  is entailed by the
store and the store is ``consistent enough'' with respect to the threshold
$a$ in the next
$t$ time-units: if $t>0$ and the condition on the store and
the cut level are satisfied, then the agent $A$ is evaluated (rule ${\bf Q10}$). If $t>0$
 and the condition on the cut level is not satisfied, then the agent $B$ is evaluated (rule ${\bf Q11}$). Finally if $t>0$, the condition on the cut level is satisfied, but the condition on the store is not satisfied, then the control is repeated at the
next time instant and the value of the counter $t$ is decreased
(axiom ${\bf Q12}$); note that in this case we use the label
$\omega$, since a check on the store has been performed. As shown
by axiom ${\bf Q13}$, the counter can be  decreased also by
performing a $\tau$-action: intuitively, this rule is used to model
the situation in which, even though the evaluation of the time-out
started already, another (parallel) process has the control. In
this case, analogously to the approach in \cite{BoGaMe04} and differently from the approach in \cite{BGZ00}, time
continues to elapse (via $\tau$-actions) also for the time-out
process (see also the rules {\bf Q5} and {\bf Q6} of the parallel operator). Axiom
${\bf Q14}$ shows that, if the time-out is exceeded, i.e., the
counter $t$ has reached the value of $0$, then the process $\hbox{\bf askp}_t(c)?^a A  \mathit{:} B$ behaves as $B$.
\item[Askp$_t$] The rules ${\bf Q15}$-${\bf Q19}$ in Figure~\ref{t2} are similar to rules ${\bf Q10}$-${\bf Q14}$ described before, with the exception that here a finer (pointwise)
  threshold $\phi$ is compared to the store $\sigma$, analogously to what happens with the $\tell$ and $\ask$ agents.
\end{description}

In the following we provide the definition for the observables of the language, which are clearly based only on $\omega$-actions.

\begin{definition}[Observables for {\itshape tsccp-i}]
Let $P= F\text{.}A$ be a {\itshape tsccp-i} process. We define
\[{\cal O}_{io}^{i}(P) = \{
\gamma \Downarrow_{Fv(A)} \mid \la A, \bar{\1} \ra\stackrel{\omega}{\Rrarrow} ^*\la
{\bf Success}, \gamma\ra\},
\]
\end{definition}
where ${\bf Success}$ is any agent that contains only occurrences
of the agent ${\bf success}$ and of the operator $\parallel$.

\section{An Execution Timeline for a {\itshape tsccp-i} Parallel Agent}\label{sec:timeline}
In this section we show a timeline for the execution of three {\itshape tsccp-i} agents in parallel. We consider the three soft constraints shown in Figure~\ref{fig:mpconstexample} and the \emph{Weighted} semiring $\langle
\mathbbm{R}^{+} \cup \{+\infty\},min,+,+\infty,0 \rangle$~\cite{bistarellibook,jacm}. Our parallel agent is defined by:
\[\begin{array}{ll}
  A_1 :: & {\bf askp}_5(c_3)?^{+\infty} (\tell(c_1) \rightarrow^{+\infty} \success)  \mathit{:}  (\success) \\
  A_2 :: & \tell(c_1)\rightarrow^{+\infty} \success \\
  A_3 ::& \tell(c_2)\rightarrow^{+\infty} \success\text{.}
\end{array}
\]

Their concurrent evaluation in the $\bar{0}$ empty store is shown in Figure~\ref{figure:timeline4}. \comment{At $t=0$ the store does not entail $c_3$, thus the ${\bf askp}_5$ of agent $A_1$ can only make  $\tau$-transitions (rule {\bf Q12} in Figure~\ref{t2}), waiting for the satisfaction of its guard (i.e. $c_3$) or for the elapsing of $5$ time-units.} At $t=0$ and $t=1$ the agent $A_1$ can  make a $\tau$-transition (rule {\bf Q13} in Figure~\ref{t2}), waiting for  the elapsing of $1$ time-unit. This can be done in parallel with a single other $\omega$-action: therefore, the $\tell(c_1)$ of agent $A_2$, and the  $\tell(c_2)$ of agent $A_3$ cannot run in parallel at the same time, since they are both $\omega$-actions. In the execution shown in Figure~\ref{figure:timeline4}, $A_2$ is executed before $A_3$ (also the opposite is possible, depending on the scheduling), leading to the store $\sigma= c_1 \otimes c_2 = c_3$. At $t=2$, the guard of ${\bf askp}_5$ in agent $A_1$ is enabled since $\sigma \ent c_3$ and, therefore, rule {\bf Q10} in Figure~\ref{t2} is executed. Finally, at $t=3$ the $\tell(c_1)$ action of agent $A_1$ is executed as the last action, and at $t=4$ we have $\langle  \success \parallel  \success \parallel \success, c_1\otimes c_2 \otimes c_1 \rangle$.

\begin{figure}[h]
\centering
\includegraphics[scale=1.1]{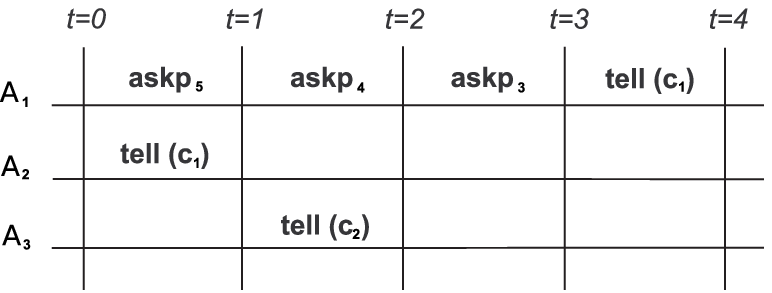} 
\caption{A timeline for the execution of $A_1 \parallel A_2 \parallel A_3$.}
\label{figure:timeline4}
\end{figure}

\section{Denotational Semantics for {\itshape tsccp-i}}\label{sec:compdensemtsccpi}

In this section we define a denotational characterization of the
operational semantics for {\itshape tsccp-i}. Differently from the denotational semantics for the maximal parallelism version presented in Section.~\ref{sec:mpcompdensem}, here for computational states we consider triples
rather than pairs, as $\omega$-actions have to be distinguished
from $\tau$-actions. This difference leads to a different technical
development.

Our denotational model for {\itshape tsccp-i} associates with a process a set of timed
reactive sequences of the form $\langle \sigma_1,\gamma_1, \xi_1\rangle
\cdots \langle \sigma_n,\gamma_n, \xi_n\rangle \langle \sigma
,\sigma,\omega\rangle$. Any triple $\langle
\sigma_i,\gamma_i, \xi_i\rangle$ represents a reaction (a computation step) of the given
process at time $i$: intuitively, the process transforms the global
store from $\sigma_i$ to $\gamma_i$ by performing a
transition step labeled by $\xi_{i}$ or, in other words, $\sigma_i$
is the assumption on the external environment, $\xi_{i}$ is the label of the
performed step while $\gamma_i$ is
the contribution of the process itself (which entails always the
assumption). The last pair denotes a ``stuttering step'', in which
the agent ${\bf Success}$ has been reached. In the following we will assume that each timed
reactive sequence $\langle \sigma_1,\gamma_1, \xi_1\rangle \cdots \langle
\sigma_{n-1},\gamma_{n-1},\xi_{n-1}\rangle\langle \sigma_n,\sigma_n, \omega\rangle$
satisfies the following condition: $\gamma_i \vdash \sigma_i
\hbox{ and } \ \sigma_j \vdash \gamma_{j-1},  $ for any $i\in
[1,n-1]$ and $j\in [2,n]$.

The basic idea underlying the
denotational model then is that, differently from the
operational semantics, inactive processes can always make a
$\tau$-step, where an inactive process is either suspended
(due to the absence of the required constraint in the store) or it is a non-scheduled component of a parallel construct. These additional
$\tau$-steps, which represent time-elapsing and are needed  to
obtain a compositional model in a simple way, are then added to
denotations as triples of the form $\la \sigma,\sigma,\tau \ra$.
For example, the denotation of the process
$\hbox{\tell}(c) \rightarrow^a {\bf success}$ contains all the reactive
sequences that have, as first element, a triple  $\la \sigma ,\sigma \otimes c,\omega\ra$ for any possible initial
store $\sigma$ with $(\sigma \otimes
c)\Downarrow_{\emptyset} \not< a$, as these represent the action of adding the constraint  $c$
to the current store. However, such a denotation contains also
sequences where the triple $\la \sigma ,\sigma \otimes c,\omega\ra$ (still with $(\sigma \otimes
c)\Downarrow_{\emptyset} \not< a$) is preceded by
a finite sequence of triples of the form
$\la \sigma_1 ,\sigma_1 ,\tau \ra \la \sigma_2 ,\sigma_2 ,\tau \ra \ldots \la \sigma_n ,\sigma_n ,\tau \ra$.
Such a sequence represents time-elapsing while the
process is inactive because some other parallel process is
scheduled.

The set of all reactive sequences for
{\itshape tsccp-i} process
is denoted by ${\cal S}_i$, its
typical elements by $s,s_1\ldots$, while sets of reactive
sequences are denoted by $S,S_1\ldots$ and $\varepsilon$ indicates
the empty reactive sequence. The  operator $\cdot$ denotes the
operator that concatenates these sequences.

\subsection{Compositionality of the Denotational Semantics  for {\itshape tsccp-i} Processes}

As in Section \ref{sec:mpcompdensem} for  the {\itshape tsccp} version, we now introduce a denotational
semantics ${\cal D} (F\text{.}A)(e)$ which is
compositional by definition and where, for technical reasons, we
represent explicitly the environment $e$ which associates a
denotation to each procedure identifier. More precisely, assuming
that {\itshape Pvar} denotes the set of procedure identifier,
${\it Env}_i={\it Pvar}\rightarrow {\cal P}({\cal S}_i)$, with typical
element $e$, is the set of {\itshape environments}. Analogously to Section \ref{sec:mpcompdensem}, given $e \in
{\it Env}_i$, $p \in {\it Pvar}$  and $f \in {\cal P}({\cal S}_i)$, we
denote by $e'=e\{f/p\}$ the new environment such that $e'(p)=f$
and $e'(p')= e(p')$ for each procedure identifier $p'\neq p$.

Before defining formally the denotational semantics, we need to define the operators $\bar{tell}$, $\bar{\sum}$,
$\bar{\parallel}$, $\bar{askp}$ and $\bar{\exists} x$, analogous to those given
in Section~\ref{sec:mpcompdensem} for the maximal parallelism language.

\begin{definition}[Semantic operators for {\itshape tsccp-i}]\label{defsemoperators}
Let $S,S_i$ be sets of reactive sequences,
 $c,c_i$ be constraints. Moreover let $\succ_i$ be either of
 the form $\succ^{a_i}$ or $\succ_{\phi_i}$, defined as in Definition \ref{def:servizio}.
Then we define the operators $\bar{tell}$, $\bar{\sum}$,
$\bar{\parallel}$, $\bar{askp}$ and $\bar{\exists} x$ as
follows:

\noindent {\bf The (valued) tell operator} $\bar{tell}^a: {\cal C} \times \wp({\cal S}_i) \rightarrow
\wp({\cal S}_i)$ ($\bar{tell}_\phi: {\cal C} \times \wp({\cal S}_i) \rightarrow
\wp({\cal S}_i)$) is the least function
(w.r.t. the ordering induced by $\subseteq$) which satisfies the
following equation
\[
\begin{array}{ll}
\bar{tell}^a(c,  S) = &
        \{ s  \in {\cal S}_i \mid
        s = \la \sigma ,\sigma\otimes c,\omega\ra \cdot s',
        \ \sigma\otimes c\Downarrow_{\emptyset} \not<a  \hbox{ and }  s' \in S\
        \}\ \  \cup \\
        &  \{ s  \in {\cal S}_i \mid
        s = \la \sigma ,\sigma,\tau \ra \cdot s' \hbox{ and } s' \in
        \bar{tell}^a(c,  S) \ \}\textit{.}
\end{array}
\]

\[
\begin{array}{ll}
\bar{tell}_{\phi}(c,  S) = &
        \{ s  \in {\cal S}_i \mid
        s = \la \sigma ,\sigma\otimes c,\omega\ra \cdot s',
        \ \sigma\otimes c \not\sqsubset
    \phi  \hbox{ and }  s' \in S\
        \}\ \  \cup \\
        &  \{ s  \in {\cal S}_i \mid
        s = \la \sigma ,\sigma,\tau \ra \cdot s' \hbox{ and } s' \in
        \bar{tell}_{\phi}(c,  S) \ \}\textit{.}
\end{array}
\]

\noindent
{\bf The guarded choice} The semantic choice operator \\
$
{\bar{\sum}} _{i=1}^n : ({\cal C} \times \wp({\cal S}_i)) \times \cdots \times ({\cal C} \times \wp({\cal S}_i)) \rightarrow
\wp({\cal S}_i)$ is the least function which satisfies the
following equation:
\[
{\bar{\sum}} _{i=1}^n c_i \succ_i \, S_i =
\begin{array}[t]{l}
\{ s\in {\cal S}_i \mid
\begin{array}[t]{l}
s = \la \sigma,\sigma, \omega \ra \cdot s', \\
\sigma\succ_h\,  c_h \hbox { and } s'\in S_h \hbox{
for an } h\in[1,n] \ \}
\end{array}
\\
\cup
\\
\{ s \in {\cal S}_i \mid
\begin{array}[t]{l}
s =\langle \sigma,\sigma ,\tau\rangle \cdot s'  \hbox{ and } \ s' \in {\bar{\sum}} _{i=1}^n c_i \succ_i \, S_i \}\textit{.}
\end{array}
\end{array}\]

\noindent {\bf Parallel Composition}. Let $\bar{\parallel}
\in{\cal S}_i\times {\cal S}_i\rightarrow {\cal S}_i$ be the
(commutative and associative) partial operator defined by
induction on the length of the sequences as follows:
\[
\begin{array}{l}
\la \sigma,\sigma,\omega\ra     \bar{\parallel} \la \sigma,\sigma,\omega \ra = \la \sigma,\sigma,\omega  \ra \\
\la \sigma,\sigma',x \ra \cdot s \bar{\parallel} \la \sigma,\sigma,\tau \ra \cdot
s' = \la \sigma,\sigma,\tau \ra \cdot s' \bar{\parallel} \la \sigma,\sigma',x \ra
\cdot s  = \la \sigma,\sigma',x \ra \cdot (s\bar{\parallel}  s'),
\end{array}
\]
where $x \in \{\omega,\tau\}$.\\
We define the operator $S_1\bar{\parallel} S_2$ on sets as the
image of ${\cal S}_i\times {\cal S}_i$ under the above operator.

\noindent
{\bf The (valued) askp operator} $\bar{askp}(t)^a:
{\cal C} \times \wp({\cal S}_i)\times \wp({\cal S}_i) \rightarrow
\wp({\cal S}_i)$ ($\bar{askp(t)}_\phi:
{\cal C} \times \wp({\cal S}_i)\times \wp({\cal S}_i) \rightarrow \wp({\cal S}_i)$),
with $t>0$, is defined as:
\[\begin{array}{lll}
  \bar{askp}(t)^a(c,  S_{1} , S_{2})=
 & \{ s  \in {\cal S}_i \mid & s = \la \sigma,\sigma, \omega \ra \cdot s' \hbox{ and}\\
&& \hbox{either }
 \sigma \succ^{a} c  \mbox{ and } s \in S_{1}\\
&& \hbox{or } \sigma\Downarrow_{\emptyset} < a   \mbox{ and } s \in S_{2} \}\ \ \cup \\
&\{ s  \in {\cal S}_i \mid
        & s = \la \sigma ,\sigma, x \ra \cdot s' , \,
        s' \in \bar{askp}(t-1)^a(c,  S_{1} , S_{2})  \\
        &&  \hbox{and either } x=\tau\\
       &&  \hbox{or } x=\omega, \,  \sigma \not \ent c  \mbox{ and }  \sigma\Downarrow_{\emptyset} \not< a \ \}\textit{.}
\end{array}
\]

\[
\begin{array}{lll}
\bar{askp}(t)_\phi(c,  S_{1} , S_{2})=&
\{ s  \in {\cal S}_i \mid & s = \la \sigma,\sigma', \omega \ra \cdot s' \hbox{ and}\\
&& \hbox{either }
 \sigma \succ_{\phi} \,c    \mbox{ and } s \in S_{1}\\
 && \hbox{or } \sigma \sqsubset \phi  \mbox{ and } s \in S_2 \}\ \ \cup \\
&\{ s  \in {\cal S}_i \mid
        & s = \la \sigma ,\sigma, x \ra \cdot s' , \,
        s' \in \bar{askp}(t-1)_\phi(c,  S_{1} , S_{2}) \\
        &&  \hbox{and either } x=\tau\\
       &&  \hbox{or } x=\omega, \,  \sigma \not \ent c  \mbox{ and }  \sigma
\not \sqsubset \phi  \ \}\textit{.}
\end{array}
\]

\noindent {\bf The (valued) askp operator} $\bar{askp}(0)^a:
{\cal C} \times \wp({\cal S}_i)\times \wp({\cal S}_i) \rightarrow
\wp({\cal S}_i)$ ($\bar{askp}(0)_\phi:
{\cal C} \times \wp({\cal S}_i)\times \wp({\cal S}_i) \rightarrow \wp({\cal S}_i)$) is the least function which satisfies the
following equation
\[
\begin{array}{lll}
 \bar{askp}(0)^a(c,  S_{1} , S_{2})=
 &\{ s  \in {\cal S}_i \mid
        & \hbox{either }
        s = \la \sigma, \sigma, \omega \ra \cdot s'
        \hbox{ and } s' \in S_{2} \\
        && \hbox{or } s = \la \sigma ,\sigma,\tau \ra \cdot s'\\
       &&\mbox{and } s' \in \bar{askp}(0)^a(c,  S_{1} , S_{2}) \ \}\textit{.}
\end{array}
\]

\[
\begin{array}{lll}
 \bar{askp}(0)_\phi(c,  S_{1} , S_{2})=
 &\{ s  \in {\cal S}_i \mid
        & \hbox{either }
        s = \la \sigma, \sigma, \omega \ra \cdot s'
        \hbox{ and } s' \in S_{2} \\
        && \hbox{or } s = \la \sigma ,\sigma,\tau \ra \cdot s' \\
       && \mbox{and } s' \in \bar{askp}(0)_\phi(c,  S_{1} , S_{2} ) \ \}\textit{.}
\end{array}
\]

\noindent
{\bf The hiding operator}
The semantic hiding operator can be defined as
follows:
\[
\begin{array}{ll}
  {\bf \bar{ \exists} } x S =  \{ s \in {\cal S}_i \mid & \mbox{\rm
  there exists $ s' \in S$ such that $s
= s'[x/y]$ with $y$ new } \}
\end{array}
\]
where $s'[x/y]$ denotes the sequence obtained from $s'$ by
replacing the variable $x$ for the variable $y$ that we assume to
be new.\footnote{As before, we assume that each time that we
consider a new applications of the operator  ${\bf \bar{
\exists} }$ we use a new, different $y$.}

\end{definition}

It is immediate  to see that the previous semantic operators are
well defined, that is, the least function which satisfies the
equations actually exists and can be obtained by a standard
fix-point construction.
The $\bar{tell}$, $\bar{\sum}$,
$\bar{\parallel}$, $\bar{askp}$ and $\bar{\exists} x$
operators have the expected definition, including the mentioned
addition of $\tau$-steps.

In the semantic parallel operator (acting on sequences) we
require that at each point of time at most one $\omega$-action is
present and the two arguments of the operator agree with respect
to the contribution of the environment (the first component of
the triple). We also require that the two arguments have the
same length (in all other cases the parallel composition is
assumed being undefined): this is necessary to reflect the
passage of time since the $i-th$ element of any sequence
corresponds to the given processes action on the $i-th$ time
step. Even though we merge point-wise  sequences of the same
length, this models an interleaving approach for
$\omega$-actions, because of the previously mentioned addition
of $\tau$-steps to denotations.
Concerning the semantic
choice operator, a sequence in $\bar{\sum} _{i=1}^n c_i
\succ_i \,S_i$ consists of an initial period of waiting
for a store which satisfies one of the guards. During this waiting period, only the environment is active by
producing the constraint $\sigma$, while the process itself
generates the stuttering steps $\langle \sigma,\sigma, \tau\rangle$.
When the store is strong enough to satisfy a guard, that is to entail a
$c_h$ and to satisfy the condition on the cut level, then
the resulting sequence is  obtained by adding $s'\in S_h$ to
the initial waiting period.
\comment{In the semantic parallel operator defined on sequences we require
that the two arguments of the operator agree at each point of time
with respect to the contribution of the environment (the
$\sigma_i$'s) and that they have the same length (in all other
cases the parallel composition is assumed being undefined).}

We can define the denotational semantics ${\cal D}$ as follows.
Here,
$Process_i$ denotes the set of {\itshape tsccp-i} processes.

\begin{definition}[Processes Semantics]\label{def:ira}
We define the semantics ${\cal D}\in {\it Process_i}\rightarrow{\cal P}({\cal S}_i)$
is the least function with respect to the ordering induced by the set-inclusion, which satisfies the equations in Figure~\ref{densem2}
\end{definition}

\begin{figure}
\centering
\scalebox{0.9}{
\begin{tabular}{lll}

\mbox{ }&\mbox{ }&\\

\mbox{{\bf F1}} & ${\cal D}(F\text{.}{\bf success})(e) = \{ \la
\sigma_1, \sigma_1, \tau\ra \la \sigma_2, \sigma_2,\tau\ra  \cdots \la
\sigma_n, \sigma_n,\omega\ra \in {\cal S}_i
\mid n\geq 1\}$&\\

\mbox{ }&\mbox{ }&\\

\mbox{{\bf F2}} & ${\cal D}(F\text{.}\hbox{\tell}(c) \rightarrow^{a} A)(e) =  \bar{tell}^a (c,  {\cal D}(F\text{.}A)(e))$&\\

\mbox{ }&\mbox{ }&\\

\mbox{{\bf F3}} & ${\cal D}(F\text{.}\hbox{\tell}(c) \rightarrow_{\phi} A)(e) =  \bar{tell}
 _{\phi}(c,  {\cal D}(\it F\text{.}A) (e))$&\\

\mbox{ }&\mbox{ }&\\

\mbox{{\bf F4}} & ${\cal D}({\it F\text{.}\sum_{i=1}^n }{\it \ask(c_i)
\rightarrow_i  A_i})(e) =
{\tilde{\sum} } _{i=1}^n c_i \succ_i \,{\cal D}( {\it F\text{.}A_i} )(e)$&\\

\mbox{ }&\mbox{ }&\\

\mbox{{\bf F5}} & ${\cal D}( {\it F\text{.}\hbox{\bf askp}_t(c)?^a  A  \mathit{:} B})(e)=  \bar{askp}(t)^a(c,{\cal D}({\it F\text{.}A})(e),{\cal D}( {\it F\text{.}B})(e))$&\\

\mbox{ }&\mbox{ }&\\

\mbox{{\bf F6}} & ${\cal D}( {\it F\text{.}\hbox{\bf askp}_t(c)?_{\phi} A  \mathit{:} B})(e)= \bar{askp}(t)_\phi(c,{\cal D}({\it F\text{.}A})(e),{\cal D}( {\it F\text{.}B})(e))
$&\\

\mbox{ }&\mbox{ }&\\

\mbox{{\bf F7}} & ${\cal D}( {\it F\text{.}A\parallel B})(e) = {\cal D}( {\it
F\text{.}A})(e) \ \bar{{\it  \parallel } }
\ {\cal D}( {\it F\text{.}B})(e)$&\\

\mbox{ }&\mbox{ }&\\

\mbox{{\bf F8}} &
${\cal D}( {\it F\text{.}\exists x A})(e) = \bar{\exists}x {\cal D}( {\it F\text{.}A} )(e)$&\\

\mbox{ }&\mbox{ }&\\

\mbox{{\bf F9}} &
${\cal D}( {\it F\text{.}p(x)}) (e) = \mu \Psi^i \; \; \text{where}  \; \;
\Psi^i(f)={\cal D}( F\setminus\{p\}\text{.} ask(\bar{\1}) \rightarrow
A) (e\{f/p\})$, &\\ & $p(x)::A \in F$&\\

\mbox{ }&\mbox{ }&\\
\end{tabular}}
\caption{The semantics  $D( F\text{.}A)$(e) for {\itshape tsccp-i}.}\label{densem2}
\end{figure}

Also ${\cal D}$ is well defined and can be obtained by a fix-point
construction. To see this, let us define an interpretation as a
mapping $I:Process_i\rightarrow \wp({\cal S}_i)$. Then let us denote
by ${\cal I}$ the cpo of all the interpretations (with the
ordering induced by $\subseteq$).  To the
equations in Figure~\ref{densem2}, we can then associate a monotonic (and continuous)
mapping ${\cal F}: {\cal I} \rightarrow {\cal I}$ defined by the
equations of Figure~\ref{densem2}, provided that we replace the
symbol ${\cal D}$ for ${\cal F}(I)$, we delete the environment $e$ and that we replace equation
${\bf F9}$ for the following one: ${\cal F}(I)( {\it F\text{.}p(x)})=
I( F\text{.} ask(\bar{\1}) \rightarrow
A)\textit{.}$

Then, one can easily prove that a function satisfies the equations
in Figure~\ref{densem2} iff it is a fix-point of the function ${\cal
F}$. Because this function is continuous (on a cpo), well known
results ensure us that its least fix-point exists and it equals
${\cal F}^\omega$, where the powers are defined as follows: ${\cal
F}^0 = I_0$ (this is the least interpretation which maps any
process  to the empty set); ${\cal F}^ n  = {\cal F} ({\cal
F}^{ n-1})$ and ${\cal F}^ \omega = lub \{ {\cal F}^ n
| n\geq 0\}$ (where lub is the least upper bound on the cpo ${\cal
I}$).

\subsection{Correctness of the Denotational Semantics  for {\itshape tsccp-i} Processes}\label{sec:cor}

As for the correctness of the denotational semantics presented in Section~\ref{sec:mpcorrectness}, at each step, the assumption on the current store must be
equal to the store produced by the previous step. In other words,
for any two consecutive steps $\la \sigma_i,\sigma'_i,x_{i}\ra \la
\sigma_{i+1},\sigma'_{i+1}, x_{i+1}\ra$ we must have $\sigma'_i =
\sigma_{i+1}$.
Furthermore, triples containing $\tau$-actions do not correspond
to observable computational steps, as these involve
$\omega$-actions only.

\begin{definition}[Connected  Sequences in {\itshape tsccp-i}]
Let $s=\la \sigma_1,\sigma'_1, x_{1}\ra \la \sigma_2,\sigma'_2, x_{2}\ra \cdots
\la \sigma_{n},$ $\sigma_{n},\omega\ra$ be a reactive sequence. We say that
$s$ is connected if $\sigma_1= \bar{\1} $, $\sigma_i = \sigma'_{i-1}$ and $x_j=\omega$
for each $i, j$, $2\leq i\leq n$ and $1\leq j\leq n-1$.
\end{definition}

According to the previous definition, a sequence is connected if
all the information assumed on the tuple space is produced by
the process itself and only
$\omega$-actions are involved. To be defined as
connected, a sequence must also have $\bar{\1}$ as the initial
constraint.  A connected sequence represents a
{\itshape tsccp-i} computation, as it will be proved in the
remaining of this section.

In order to prove the correctness of the denotational semantics,
we use a modified transition system $T'$, where inactive (either
suspended or not scheduled) processes can perform $\tau$-actions.
When considering our notions of observables, we can prove that
such a modified transition system is equivalent to the previous
one and agrees with the denotational model.

The new transition system $T'$ is obtained from the one in
Figure~\ref{t2} by deleting rule {\bf Q6} and by adding the rules
{\bf Q0'}, {\bf Q1'}, {\bf Q2'}, {\bf Q3'},  {\bf Q4'}, {\bf Q7'}, {\bf Q8'}, {\bf Q14'} and {\bf
Q19'}, contained in Figure~\ref{rulestau}. We denote by
$\Rightarrow$ the relation defined by $T'$.

   \begin{figure}
    \begin{center}
    \scalebox{0.9}{
\begin{tabular}{llll}
&\mbox{   }&\mbox{   } &\mbox{   }
\\
\mbox{\bf Q0'}& $
\begin{array}{l}
\la \hbox{success}, \sigma \ra \stackrel{\tau}{\Rightarrow} \la \hbox{success},
\sigma \ra
\end{array}$\ \ \ & \bf{success}&
\\
&\mbox{   }&\mbox{   } &\mbox{   }
\\
\mbox{\bf Q1'}& $
\begin{array}{l}
\la \hbox{\tell}(c) \rightarrow^{a} A, \sigma \ra \stackrel{\tau}{\Rightarrow}
\la \hbox{\tell}(c) \rightarrow^{a} A, \sigma \ra
\end{array}$\ \ \ & \bf{V-Telll}&
\\
&\mbox{   }&\mbox{   } &\mbox{   }
\\
\mbox{\bf Q2'}& $
\begin{array}{l}
\la \hbox{\tell}(c)\rightarrow_{\phi} A, \sigma \ra \stackrel{\tau}{\Rightarrow} \la
\hbox{\tell}(c)\rightarrow_{\phi} A, \sigma \ra
\end{array}$ & \bf{Tell} &
\\
&\mbox{   }&\mbox{   } &\mbox{   }
\\
\mbox{\bf Q3'}& $
\begin{array}{l}
\la \hbox{\ask}(c) \rightarrow^{a} A, \sigma \ra \stackrel{\tau}{\Rightarrow}  \la\hbox{\ask}(c) \rightarrow^{a} A, \sigma \ra
\end{array}$\ \ \ & \bf{V-ask}&
\\
&\mbox{   }&\mbox{   } &\mbox{   }
\\
\mbox{\bf Q4'}& $
\begin{array}{l}
\la \hbox{\ask}(c) \rightarrow_{\phi} A, \sigma  \ra \stackrel{\tau}{\Rightarrow}
\la\hbox{\ask}(c) \rightarrow_{\phi} A, \sigma \ra
\end{array}$ & \bf{Ask}&
\\
&\mbox{   }&\mbox{   } &\mbox{   }
\\
\mbox{\bf Q7'}& $
\begin{array}{l}
\la \Sigma_{i=1}^{n}E_i , \sigma \ra \stackrel{\tau}{\Rightarrow}\la \Sigma_{i=1}^{n}E_i , \sigma \ra
\end{array}$ & \bf{Nondet}&
\\
&\mbox{   }&\mbox{   }&
\\
\mbox{\bf Q8' }
& $\la p(x),\sigma\ra\stackrel{\tau}{\Rightarrow}\la A,
\sigma\ra \ \ \ \ {\it p(x) :: A \in F}$
&\bf{P-call}&
\\
&\mbox{   }&\mbox{   } &\mbox{   }
\\
\mbox{\bf Q14'}&$
\la \hbox{\bf askp}_0(c)?^a A  \mathit{:} B,\sigma\ra\stackrel{\tau}{\Rightarrow} \la \hbox{\bf askp}_0(c)?^a A  \mathit{:} B,\sigma\ra
$& \bf{V-askp5}&
\\
&\mbox{   }&\mbox{   }&
\\
\mbox{\bf Q19'}
& $
\la \hbox{\bf askp}_0(c)?_{\phi} A  \mathit{:} B,\sigma\ra\stackrel{\tau}{\Rightarrow} \la \hbox{\bf askp}_0(c)?_{\phi} A  \mathit{:} B,\sigma\ra
$& \bf{Askp5}
&\\
&\mbox{   }&\mbox{   }&
\\
\end{tabular} }
    \end{center}

  \caption{The $\tau$-rules for {\itshape tsccp-i}.}\label{rulestau}
\end{figure}

The observables induced by the transition system $T'$ are
formally defined as follows.

\begin{definition}
Let $P= F\text{.}A$ be a {\itshape tsccp-i} process. We define
\[{\cal O}_{io}^{i'}(P) = \{
\gamma \Downarrow_{Fv(A)} \mid \la A, \bar{\1} \ra \stackrel{\omega}{\Rightarrow}\,\!\!^*\la
{\bf Success}, \gamma\ra\},
\]
\end{definition}
where ${\bf Success}$ is any agent which contains only occurrences
of the agent ${\bf success}$ and of the operator $\parallel$.

Lemma~\ref{lem:opequivalent} shows that the modified transition
system agrees with the original one when considering our notion of
observables.

We first need some definitions and technical lemmata. In the following, given two agents $A$ and $B$, we say that  $A \simeq B$ if and only if
$B$ is obtained from $A$ by replacing an agent of the form $\exists x A_1$ in $A$ with $A_1 [x/y]$, where $y$ is new in $A$.
$\approx$ denotes the reflexive and transitive closure of $ \simeq$. The following lemmata hold.

\begin{lemma}\label{lem:unf}
Let $F\text{.}A$ and $F\text{.}B$ be {\itshape tsccp-i} processes such that $ A \approx B$. Then for each store $\sigma$ and for $x \in \{ \omega, \tau
\}$
   $$\la F\text{.}A,\sigma\ra \stackrel{x}{\Rrarrow}
    \la F\text{.}C,\sigma'\ra \mbox{ if and only if }
    \la F\text{.}B,\sigma\ra \stackrel{x}{\Rrarrow}
    \la F\text{.}C,\sigma'\ra\textit{.}$$

\end{lemma}
\begin{proof}
    The proof is immediate, by using rule {\bf Q9} and by a straightforward inductive argument.
 \end{proof}

From the above Lemma we derive the following corollary:

\begin{corollary}\label{cor:unf}
Let $F\text{.}A$ and $F\text{.}B$ be {\itshape tsccp-i} processes such that $ A \approx B$.
 Then for each store $\sigma$, $\la F\text{.}A, \sigma \ra\stackrel{\omega}{\Rrarrow}\,\! \!^*\la {\bf
Success}, \gamma\ra$ if and only if $\la F\text{.}B, \sigma \ra\stackrel{\omega}{\Rrarrow}\,\!\! ^*\la {\bf
Success}, \gamma\ra$.
\end{corollary}

\begin{lemma}\label{lem:cambio}
Let $P= F\text{.}A$ be a {\itshape tsccp-i} process.  Then for each store $\sigma$,
\begin{enumerate}
    \item $\la F\text{.}A, \sigma \ra \stackrel{\tau}{\Rightarrow}\la F\text{.}B, \sigma'\ra$
    if and only if $\sigma= \sigma'$ and
    \begin{itemize}
        \item either $\la F\text{.}A, \sigma \ra \stackrel{\tau}{\Rrarrow}\la F\text{.}C, \sigma\ra$  and
        $C\approx B$
        \item or $\la F\text{.}A, \sigma \ra \stackrel{\tau}{\not \Rrarrow}$ and
        $B\approx A$.
    \end{itemize}
    \item $\la F\text{.}A, \sigma \ra \stackrel{\omega}{\Rightarrow}\la F\text{.}B, \sigma'\ra$
    if and only if $\la F\text{.}A, \sigma \ra \stackrel{\omega}{\Rrarrow}\la F\text{.}C, \sigma'\ra$ and $C \approx B$.
\end{enumerate}
\end{lemma}
\begin{proof}
 \begin{enumerate}
   \item
 The proof is by induction on the complexity of the agent $A$.
  \begin{itemize}
    \item $A$ is of the form  $\bf{success}$,
    $\hbox{\tell}(c) \rightarrow^{a} A$,
    $\hbox{\tell}(c)\rightarrow_{\phi} A$,
    $\hbox{\ask}(c) \rightarrow^{a} A$,
    $\hbox{\ask}(c) \rightarrow_{\phi} A$,
    $\Sigma_{i=1}^{n}E_i$, $p(x)$,
    $\hbox{\bf askp}_0(c)?^a A  \mathit{:} B$ and
    $\hbox{\bf askp}_0(c)?_{\phi} A  \mathit{:} B$.

    The proof is immediate by observing that by the rules
    in Figure~\ref{t2}, $\la A, \sigma \ra \stackrel{\tau}{\not \Rrarrow}$ and
    by the rules in Figure~\ref{rulestau},
    $\la A, \sigma \ra  \stackrel{\tau}{\Rightarrow}\la B, \sigma' \ra $
    if and only if $\la A, \sigma \ra=\la B, \sigma' \ra $.
    \item $A$ is of the form $\hbox{\bf askp}_t(c)?^a A_1  \mathit{:} A_2$
    ($\hbox{\bf askp}_t(c)?_{\phi} A_1  \mathit{:} A_2$), with
    $t>0$.\\
    The proof is immediate since both the transition
    systems use the rule {\bf Q13} ({\bf Q18}) of Figure~\ref{t2}.
    \item If $A$ is of the form $A_1\parallel A_2$. \\
    In this case, by definition of the transition system $T'$ and by using rule
    {\bf Q5} of Figure~\ref{t2}, for each store $\sigma$,
    \[\begin{array}{l}
      \la  A_1\parallel A_2, \sigma \ra \stackrel{\tau}{\Rightarrow}
      \la  B_1\parallel B_2, \sigma' \ra \mbox{ if and only if } \\
      \la  A_1, \sigma \ra \stackrel{\tau}{\Rightarrow}
      \la  B_1, \sigma'\ra \mbox{ and }
      \la  A_2, \sigma \ra \stackrel{\tau}{\Rightarrow}
      \la  B_2, \sigma\ra
    \end{array}
    \]
    (the symmetric case is analogous and hence it is omitted). \\
    By inductive hypothesis this holds if and only if  $\sigma'=\sigma$ and for $i=1,2$
    \begin{itemize}
        \item  either $\la  A_i, \sigma \ra \stackrel{\tau}{\Rrarrow}  \la  C_i, \sigma\ra$ and $C_i \approx B_i$
        \item or $\la  A_i, \sigma \ra \stackrel{\tau}{\not \Rrarrow}$
        and $B_i \approx A_i$.
    \end{itemize}
   If there exists $i \in [1,2]$ such that
   $\la  A_i, \sigma \ra \stackrel{\tau}{\Rrarrow}  \la  C_i, \sigma\ra$
   then the thesis follows by using either rule {\bf Q5} or rule {\bf
   Q6}.\\
   Otherwise $ \la  A_1\parallel A_2, \sigma \ra  \stackrel{\tau}{\not
   \Rrarrow}$. Then the thesis follows since by the previous results
   ${B_1\parallel B_2} \approx {A_1\parallel A_2}$.

    \item $A$ is of the form $\exists x A_1$.
    By rule {\bf Q9} of Figure~\ref{t2} for each store $\sigma$,
    \[\begin{array}{l}
      \la \exists x A_1, \sigma \ra\stackrel{\tau}{\Rightarrow}\la B, \sigma' \ra \mbox{ if and only if }
       \la A_1[x/y],  \sigma
\ra\stackrel{\tau}{\Rightarrow}\la B, \sigma' \ra
    \end{array}
    \]
    By inductive hypothesis this holds if and only if $\sigma'=\sigma$
    \begin{itemize}
        \item either $\la A_1[x/y],  \sigma
\ra\stackrel{\tau}{\Rrarrow} \la C, \sigma \ra$ and $C \approx B$

        \item or $\la A_1[x/y],  \sigma
\ra\stackrel{\tau}{\not \Rrarrow}$ and $B \approx A_1[x/y]$.
    \end{itemize}
    Therefore, by using rule {\bf Q9} of Figure~\ref{t2} and since $ \exists x A_1 \approx A_1[x/y]$, we have that
    \begin{itemize}
         \item either $\la \exists x A_1,  \sigma
\ra\stackrel{\tau}{\Rrarrow} \la C, \sigma \ra$  and $C \approx B$

        \item or $\la \exists x A_1,  \sigma
\ra\stackrel{\tau}{\not \Rrarrow}$ and $B \approx \exists x A_1$
and then the thesis.
    \end{itemize}

   \end{itemize}
   \item  The proof is analogous to the previous one and hence it is
  omitted.
 \end{enumerate}
 \end{proof}

\begin{lemma}\label{lem:opequivalent}
Let $P= F\text{.}A$ be a {\itshape tsccp-i} process. Then ${\cal O}_{io}^{i'} (P) = {\cal O}_{io}^{i} (P) .$
\end{lemma}
\begin{proof}
We prove that there exists a computation $\la A, \sigma \ra\stackrel{\omega}{\Rightarrow}\,\!\! ^*\la {\bf
Success}, \gamma\ra$
if and only if there exists a computation $\la A, \sigma \ra\stackrel{\omega}{\Rrarrow} ^*\la {\bf
Success}, \gamma\ra$. Then the thesis follows by definition of ${\cal O}_{io}^{i} (P)$ and ${\cal O}_{io}^{i'} (P).$
The proof is by induction on the length of the computation $\la A, \sigma \ra\stackrel{\omega}{\Rightarrow}\,\!\! ^*\la {\bf
Success}, \gamma\ra$.
\begin{description}
  \item[$n=1)$] In this case $A= {\bf Success}$ and then the thesis.
  \item[$n>1)$] In this  case
    \[\begin{array}{lll}
      \la A, \sigma \ra\stackrel{\omega}{\Rightarrow} \,\!\!^*\la {\bf Success}, \gamma\ra &
      \mbox{iff} & \\
      \hspace*{1cm}
      \mbox{(by definition)} \\
      \la A, \sigma \ra\stackrel{\omega}{\Rightarrow} \la A_1, \sigma_1 \ra
      \mbox{ and } \la A_1, \sigma_1 \ra\stackrel{\omega}{\Rightarrow}\,\!\! ^*\la {\bf Success}, \gamma\ra &
      \mbox{iff} & \\
      \hspace*{1cm}\mbox{(by inductive hypothesis)} \\
      \la A, \sigma \ra\stackrel{\omega}{\Rightarrow} \la A_1, \sigma_1 \ra
      \mbox{ and } \la A_1, \sigma_1 \ra\stackrel{\omega}{\Rrarrow}\,\!\! ^*\la {\bf Success}, \gamma\ra &
      \mbox{iff} &  \\
      \hspace*{1cm}\mbox{(by Point 2 of Lemma~\ref{lem:cambio})} \\
      \la A, \sigma \ra\stackrel{\omega}{\Rrarrow} \la A_2, \sigma_1 \ra, \ A_2 \approx A_1
      \mbox{ and } \la A_1, \sigma_1 \ra\stackrel{\omega}{\Rrarrow} \,\!\!^*\la {\bf Success}, \gamma\ra & \mbox{iff} &\\
      \hspace*{1cm}
      \mbox{(by Corollary~\ref{cor:unf})}\\
      \la A, \sigma \ra\stackrel{\omega}{\Rrarrow} \la A_2, \sigma_1 \ra
      \mbox{ and } \la A_2, \sigma_1 \ra\stackrel{\omega}{\Rrarrow} \,\!\!^*\la {\bf Success}, \gamma\ra & \mbox{iff} &\\
      \hspace*{1cm}
      \mbox{(by definition)} \\
      \la A, \sigma  \ra\stackrel{\omega}{\Rrarrow} \,\!\!^*\la {\bf Success}, \gamma\ra.
    \end{array}
    \]

\end{description}
\end{proof}

We can now easily prove that, given our definition of ${\cal D}$, the modified transition system $T'$ agrees with the
denotational model.

\begin{theorem}\label{lem:equivalent}
For any {\itshape tsccp-i} process $P= F\text{.}A$ we have

$${\cal O}_{io}^{i'}(P) = \{ \sigma_n \Downarrow _{Fv(A)}\mid
\begin{array}[t]{ll} \mbox{\rm there exists a connected sequence } s\in
{\cal D}(P) \ \mbox{\rm such that }
\\
 s=\la \sigma_1,\sigma_2,\omega\ra \la \sigma_2, \sigma_3,\omega\ra\cdots
 \la \sigma_{n},\sigma_{n},\omega\ra \}\textit{.}
\end{array}
$$
\end{theorem}
\begin{proof} We prove by induction on the complexity of the agent $A$ that
\[
\begin{array}{ll}
  {\cal D}(P)=\{ s  \mid & s=\la \sigma_1,\sigma_1',x_1\ra \la \sigma_2,\sigma_2', x_2\ra\cdots
  \la \sigma_{n},\sigma_{n},\omega \ra, \ A_1=A, \\
  & \mbox{for } i\in [1,n-1],\ \la {A_i},\sigma_i\ra \stackrel{x_i}{\Rightarrow}\la
  {A_{i+1}},\sigma'_i\ra  \mbox{ and }  A_n ={\bf Success}\}\textit{.}
\end{array} \]
Then the proof follows by definition of ${\cal O}_{io}^{i'}(P)$.

When the {\itshape tsccp-i} $P$ is not of the form $F\text{.}B \parallel C$ the thesis
follows immediately from the close correspondence between the
rules of the transition system and the definition of the
denotational semantics.

Assume now that $P$ is of the form $F\text{.}B \parallel C$.
By definition of the denotational semantics, $s \in {\cal D}(P)$ if and only if
$s=\la \sigma_1,\sigma_1',x_1\ra \la \sigma_2,\sigma_2', x_2\ra\cdots
  \la \sigma_{n},\sigma_{n},\omega \ra$ and there exist $s' \in
{\cal D}(F\text{.}B)$ and $s'' \in {\cal D}(F\text{.}C)$,
\[\begin{array}{ll}
    s'=\la
\sigma_1,\kappa_1',x_1'\ra \la \sigma_2,\kappa_2', x_2'\ra\cdots \la
\sigma_{n},\sigma_{n},\omega \ra & \mbox{and} \\
    s''=\la \sigma_1,\kappa_1'',x_1''\ra \la
\sigma_2,\kappa_2'', x_2''\ra\cdots \la \sigma_{n},\sigma_{n},\omega \ra, &
  \end{array}
\]
such
that  for each $i\in [1,n-1]$,
\begin{eqnarray}
  \begin{array}{ll} x_i=\tau \mbox{ if and only if }
  & x'_i=x''_i=\tau \mbox{ and in this case }
  \sigma_i'=\kappa_i'=\kappa_i''=\sigma_i,  \\
  x_i=\omega \mbox{ if and only if }
  &\mbox{either } x'_i=\omega, \ x''_i=\tau, \ \kappa_i'=\sigma_i' \mbox{ and } \kappa_i''=\sigma_i \\
  & \mbox{or } x'_i=\tau, \ x''_i=\omega, \ \kappa_i''=\sigma_i' \mbox{ and }
  \kappa_i'=\sigma_i.
  \end{array}\label{24feb11}
\end{eqnarray}
By inductive hypothesis $s' \in
{\cal D}(F\text{.}B)$ and $s'' \in {\cal D}(F\text{.}C)$ if and only if
\begin{eqnarray}
\begin{array}{ll}
  \la B_i,\sigma_i\ra \stackrel{x'_i}{\Rightarrow}\la {
B_{i+1}},\kappa'_i\ra \mbox{ for } i\in [1,n-1],\  B_1=B \mbox{ and } B_n={\bf Success},  \\
\la C_i,\sigma_i\ra \stackrel{x''_i}{\Rightarrow}\la {C_{i+1}},\kappa''_i\ra \mbox{ for } i\in [1,n-1], \  C_1=C \mbox{ and } C_n={\bf Success}\text{.}
\end{array}\label{24feb12}
\end{eqnarray}

Therefore, by Rule {\bf R8} and by (\ref{24feb11}), we have that
(\ref{24feb12}) holds if and only if
\[\begin{array}{ll}
   \la B_i \parallel C_i,\sigma_i\ra \stackrel{x_i}{\Rightarrow}\la B_{i+1} \parallel
C_{i+1},\sigma_{i+1}\ra \mbox{ for } i\in [1,n-1],& \\
   B_1\parallel C_1 =B\parallel C \mbox{ and }B_n\parallel C_n ={\bf Success} &
  \end{array}
\]
and then the thesis.
\end{proof}
\medskip

Thus we obtain the following correctness result whose proof is
immediate from the previous theorems.

\begin{corollary}[Correctness of {\itshape tsccp-i}] \label{prop:correctness}
For any {\itshape tsccp-i} process $P= F\text{.}A$ we have
\vspace{-0.2cm}
$${\cal O}_{io}^{i}(P) = \{ \sigma_n \Downarrow _{Fv(A)}\mid
\begin{array}[t]{ll} \mbox{\rm there exists a connected sequence } s\in
{\cal D}(P) \ \mbox{\rm such that }
\\
 s=\la \sigma_1,\sigma_2,\omega\ra \la \sigma_2, \sigma_3,\omega\ra\cdots
 \la \sigma_{n},\sigma_{n},\omega\ra \}\textit{.}
\end{array}
$$
\end{corollary}

\section{Related Work}\label{sec:related}
By comparing this work with other  timed languages using crisp constraints (instead of soft ones as in this paper) as \cite{SJG96,saraswat2}, there are three main differences we can find out.

First, the computational model of both the languages \emph{tcc}~\cite{saraswat2} and \emph{default tcc}~\cite{SJG96} is inspired by that one of synchronous languages: each time interval is identified with the time needed for a \emph{ccp} process to terminate a computation. Clearly, in order to ensure that the next time instant is reached, the (default) \emph{ccp} program has to be always terminating; thus, it is assumed that it does not contain recursion. On the other hand, we directly introduce  a timed interpretation of the usual programming constructs of \emph{ccp} by considering the primitive \emph{ccp} constructs {\bf ask} and {\bf tell} as the elementary actions whose evaluation takes one time-unit. Therefore, in our model, each time interval is identified with the time needed for the underlying constraint system to accumulate the tells and to answer the queries (asks) issued at each computation step by the processes of the system. For the definition of our {\itshape tsccp} agents we do not need any restriction on recursion to ensure that the next time instant is reached, since at each moment there are only a finite number of parallel agents, and the next moment in time occurs as soon as the underlying constraint system has responded to the initial actions of all the current agents of the system.

A second difference relies in the transfer of information across time boundaries. In \cite{saraswat2} and \cite{SJG96},  the programmer has to explicitly transfer  the (positive) information from a time instant to the next one, by using special primitives that allow one to control the temporal evolution of the system. In fact, at the end of a time interval all the constraints accumulated and all the processes suspended are discarded, unless they are arguments to a specific primitive. On the contrary, no explicit transfer is needed in {\itshape tsccp}, since the computational model is based on the monotonic evolution of the store which is usual in {\itshape ccp}.

A third relevant difference is in \cite{saraswat2} and \cite{SJG96} the authors present deterministic languages while our language allows for nondeterminism. These three differences also hold between  \cite{saraswat2} or \cite{SJG96}, and the original crisp version of the language, i.e., {\itshape tccp}~\cite{BGM00}.

In \cite{olarte}, the authors generalize the model in \cite{saraswat2} in order to extend it with temporary parametric {\bf ask} operations. Intuitively, these operations behave as persistent parametric asks during a time-interval, but may disappear afterwards. The presented extension goes in the direction of better modeling \emph{mobile systems} with the use of private channels between the agents. However, also the agents in \cite{olarte} show a deterministic behavior, instead of our not-deterministic choice.

Other timed extension of concurrent constraint programming have
been proposed in~\cite{NV02,PV01}, however these languages,
differently from {\itshape tsccp}, do not take into account
quantitative aspects; therefore, this achievement represents a
very important expressivity improvement with respect to related works.
These have been considered by Di Pierro and Wiklicky, who have
extensively studied probabilistic {\itshape ccp} (see for
example~\cite{DPW98}). This language provides a construct for
probabilistic choice which allows one to express randomness in a
program, without assuming any additional structure on the
underlying constraint system. This approach is therefore deeply
different from ours. More recently, stochastic {\itshape ccp} has been
introduced in~\cite{Bo06} to model biological systems. This
language is obtained by adding a stochastic duration to the ask
and tell primitives, thus it differs from our solutions.

In literature we can find other proposals that are related to tuple-based kernel-languages instead of a constraint store, as \emph{KLAIM}~\cite{dNFP98} (\emph{A Kernel Language for Agents Interaction and Mobility}) or \emph{SCEL}~\cite{scel} (\emph{Software Component Ensemble Language}) for instance. These languages are designed to study different properties of systems, as mobility and autonomicity of modeled agents. Their basic specification do not encompass time-based primitives, while mobility features are not present in any of the constraint-based languages reported in this section. The purpose of our language is to model systems where a level of preference and time-sensitive primitives (as a timeout) is required: a good example is represented by agents participating to an auction, as the example given in Section~\ref{sec:mpauctionexample}.

In general, since semiring-based soft constraints allow one to express several
quantitative features, our proposal provides a framework which can
be instantiated to obtain a variety of specific extensions of
{\itshape ccp}.

\section{Conclusion and Future Work}\label{sec:conclusions}
We have presented the {\itshape tsccp} and {\itshape tsccp-i} in order to join
together the expressive capabilities of soft constraints and
timing mechanisms in a new programming framework. The agents
modeled with these languages are able to deal with time and
preference-dependent decisions that are often found during
complex interactions. An application scenario can be represented by different entities
that need to negotiate generic
resources or services, as, for instance, during an auction process. Mechanisms as timeout and interrupt may model the wait for pending conditions or the triggering of some new
events. All the {\itshape tsccp}  and {\itshape tsccp-i} rules have been
formally described by a transition system and, then, also with a
denotational characterization of the operational semantics
obtained with the use of {\itshape timed reactive sequences}. The
resulting semantics has been proved to be compositional and
correct.

About future work, a first improvement of the presented languages can be the inclusion of a
{\itshape fail} agent in the syntax given in
Definition~\ref{def:tscclanguage}  and Definition~\ref{def:tsccilanguage}, and a semantics for the transition rules that lead to a failed computation, in case the guard on the transition rule cannot be enforced due to the preference of the store. In fact, the transition systems we have
defined consider only successful computations. If this could be a
reasonable choice in a don't know interpretation of the language
it will lead to an insufficient analysis of the behavior in a
pessimistic interpretation of the indeterminism. 

At last, we would like to consider other time management
strategies (as the one proposed in~\cite{timeVal}), and to study
how timing and non-monotonic constructs~\cite{nonmono} can be
integrated together.

\bibliography{bibliotsccp}

\end{document}